\documentclass[bj,numbers,noshowframe]{imsart}

\usepackage{amsthm,amsmath}
\usepackage{amssymb}
\usepackage{bbm}
\usepackage{enumerate}
\usepackage{tikz}
\usepackage{algorithm}
\usepackage{algorithmic}
\RequirePackage[colorlinks,citecolor=blue,urlcolor=blue]{hyperref}
\RequirePackage[numbers]{natbib}
\usepackage{breakcites}
\usepackage{todonotes}

\arxiv{arXiv:1601.08057}

\startlocaldefs

\newcommand{\X}{\mathbf{X}}
\newcommand{\R}{\mathbb{R}}
\newcommand{\B}{\mathcal{B}}
\newcommand{\e}{\varepsilon}

\DeclareMathOperator*{\esssup}{ess\,sup}

\theoremstyle{definition}
\newtheorem{defn}{Definition}[section]

\theoremstyle{plain}

\newtheorem{thm}[defn]{Theorem}
\newtheorem{prop}[defn]{Proposition}
\newtheorem{lem}[defn]{Lemma}
\newtheorem{crly}[defn]{Corollary}
\newtheorem{ex}[defn]{Example}

\theoremstyle{definition}
\newtheorem{rem}[defn]{Remark}

\endlocaldefs

\begin{document}

\begin{frontmatter}

\title{On the Geometric Ergodicity of Hamiltonian Monte Carlo}
\runtitle{On the Geo. Erg. of HMC}






\author{\fnms{Samuel} \snm{Livingstone}\corref{}\thanksref{t1}\ead[label=e1]{samuel.livingstone@ucl.ac.uk}},
\author{\fnms{Michael} \snm{Betancourt}\thanksref{t2}\ead[label=e2]{betanalpha@gmail.com}},
\author{\fnms{Simon} \snm{Byrne}\thanksref{t2}\ead[label=e3]{simonbyrne@gmail.com}}
\and
\author{\fnms{Mark} \snm{Girolami}\thanksref{t3}\ead[label=e4]{m.girolami@warwick.ac.uk}}\ead[label=e5]{m.girolami@imperial.ac.uk}

\thankstext{t1}{Supported by Xerox Research Centre Europe and EPSRC}
\thankstext{t2}{Supported by EPSRC}
\thankstext{t3}{Supported by the Royal Society and EPSRC grants EP/P020720/1, EP/J016934/3, EP/K034154/1}

\runauthor{Livingstone et al.}



\address{Department of Statistical Science, University College, Gower Street, London WC1E 6BT, United Kingdom. \\
\printead{e1,e3}}

\address{Department of Statistics, University of Warwick, Coventry, CV4 7AL, United Kingdom. \\
\printead{e2,e4}}

\address{Department of Mathematics, South Kensington Campus, Imperial College London, London SW7 2AZ, United Kingdom. \\
\printead{e5}}

\address{The Alan Turing Institute, British Library, 96 Euston Road, London NW1 2DB, United Kingdom.}

\begin{abstract} \advance\rightskip by 1em\emergencystretch=1em
We establish general conditions under which Markov chains produced by the Hamiltonian Monte Carlo method will and will not be geometrically ergodic.  We consider implementations with both position-independent and position-dependent integration times.  In the former case we find that the conditions for geometric ergodicity are essentially a gradient of the log-density which asymptotically points towards the centre of the space and grows no faster than linearly.  In an idealised scenario in which the integration time is allowed to change in different regions of the space,  we show that geometric ergodicity can be recovered for a much broader class of tail behaviours, leading to some guidelines for the choice of this free parameter in practice.
\end{abstract}



\end{frontmatter}

\section{Introduction} 
\label{sec:intro}

This paper deals with ergodic properties of Markov chains produced by the \emph{Hamiltonian} (or \emph{Hybrid}) Monte Carlo method (HMC), a technique for approximating high dimensional integrals through stochastic simulation \citep{duane1987hybrid}.  Iterative algorithms of this type are widely used in (for example) statistics and machine learning \citep{gelman2014bayesian,andrieu2003introduction}, inverse problems \citep{stuart2010inverse}, and molecular dynamics \citep{alder1959studies}.  In many of these settings a prior distribution can be constructed for an unknown quantity, and after conditioning on some observed data, Bayes' theorem gives a posterior --- to extract relevant information from this typically high-dimensional integrals must be evaluated.

A popular approach to such problems is to simulate a Markov chain whose limiting distribution is the posterior, and compute long-run averages (e.g. \citep{roberts2004general}).  Provided the chain is \emph{ergodic}, then a Law of Large Numbers exists for these.  Several \emph{Markov chain Monte Carlo} (MCMC) methods of this nature have been proposed in the literature, and many are well understood theoretically (e.g. \citep{roberts1996geometric,roberts1996exponential}).  HMC has proven an empirical success, with numerous authors noting its superior performance in a variety of settings (e.g. \citep{gelman2014bayesian}) and high performance software available for its implementation \citep{carpenter2016stan}. Comparatively few rigorous results, however, exist to justify this.  Indeed, the absence of such analysis has been noted on more than one occasion \citep{diaconis2013some,diaconis2014connections}.  The major contribution of this work is to establish general scenarios under which geometric ergodicity can and cannot be established for Markov chains produced by common HMC implementations.

Consider a Borel space $(\X,\B)$.  In this article we focus on the case $\X = \R^d$.  We define a Markov chain $(X_n)_{n \geq 0}$ on $(\X,\B)$ through an initial distribution $\delta_{x}(\cdot)$ and a family of mappings $f_\theta:\X \to \X$, indexed by $\theta$ defined on the Borel space $(\Theta,\B_{\theta})$ and with associated law $\gamma(\cdot)$ (e.g. \citep{diaconis1999iterated}).  A transition kernel $P: \X \times \B \to [0,1]$ can then be induced through the relation
\begin{equation*}
P(x,A) = \int \mathbbm{1}_A(f_\theta(x))\gamma(d\theta),
\end{equation*}
for any $A \in \B$.  Constructing a Markov chain for which some distribution of interest $\pi(\cdot)$ is \emph{invariant} is not very difficult, owing to the Metropolis--Hastings algorithm \citep{metropolis1953equation, hastings1970monte}, in which the family $\{f_\theta, \theta \in \Theta\}$ is given by
\[
f_\theta(x) := \begin{cases}
g_\xi(x) & u < \alpha(x,g_\xi(x)), \\
x & \text{otherwise},
\end{cases}
\]
where $\theta = \{\xi,u\}$ in this case, with $u \sim U[0,1]$, and $\{g_\xi, \xi \in \Xi\}$ is a family of `candidate' maps, with $\xi \sim \mu(\cdot)$.  A candidate transition kernel is induced as $Q(x,A) = \int \mathbbm{1}_A(g_\xi(x))\mu(d\xi)$ for any $A \in \B$.  If $\pi(\cdot)$ and $Q(x,\cdot)$ admit densities $\pi(x)$ and $q(x,y)$, then the `acceptance probability' $\alpha: \X \times \X \to [0,1]$ can be defined as follows.  Let $S := \{(x,y) \in \mathbf{X}^2 : \pi(x)q(x,y) >0 \}$. Then for $(x,y) \in S$ we set 
\begin{equation} \label{eqn:accept}
r(x,y) := \frac{\pi(y)q(y,x)}{\pi(x)q(x,y)},
\end{equation}
and set $r(x,y) := 0$ otherwise.  Then $\alpha(x,y) := 1 \wedge r(x,y)$.  A more general definition is given in Proposition 1 of \citep{tierney1998note}.  The resulting chain $(X_n)_{n \geq 0}$ is reversible with respect to $\pi(\cdot)$.

Simple choices for the family $\{g_\xi, \xi \in \Xi\}$ result in Markov chains which are intuitive and convenient to analyse.  In the random walk case $g_\xi(x) = x + \xi$, with $\Xi = \X$ and $\mu(\cdot)$ a centred, symmetric distribution \citep{tierney1994markov}.  For the Metropolis-adjusted Langevin algorithm (MALA) $g_\xi(x) = x + h\nabla\log\pi(x)/2 + \sqrt{h}\xi$, with $\mu(\cdot)$ a standard Gaussian measure on $\Xi = \X$, $h>0$ a constant, $\nabla$ the gradient operator and $\pi(x)$ the Lebesgue density of $\pi(\cdot)$.  The former is in some sense a naive choice, while the latter is an Euler--Maruyama scheme for the diffusion governed by $dX_t = \nabla\log\pi(X_t)dt + \sqrt{2}dW_t$, for which $\pi(\cdot)$ is invariant under suitable regularity conditions (see e.g. \citep{roberts1996exponential}).  In both cases proposals are \emph{local} (only depending on analytic information at the current point), and $x$ is combined with $\xi$ \emph{linearly}, with added complexity coming only through the (typically nonlinear) $\alpha$.  As a result, simple bounds on $\alpha$ allow stochastic stability properties such as $\pi$-irreducibility to be deduced straightforwardly, and rates of convergence for different forms of $\pi(\cdot)$ are also well-understood in both cases \citep{roberts1996geometric,roberts1996exponential}.

The HMC method can also be considered within the above framework, as outlined in \citep{betancourt2014geometric}.  The algorithm is designed to exploit the measure-preserving properties of \emph{Hamiltonian flow} (e.g. \citep{leimkuhler2004simulating}), which can be induced provided the state space for the chain is a symplectic manifold (e.g. \citep{lee2012symplectic}).  The space $\X$ can be made symplectic by doubling the dimension, introducing auxiliary \emph{momentum} variables $p$ which follow some user-specified distribution.  A Hamiltonian function can then be constructed on the resulting \emph{phase space} which preserves a distribution for $(x,p)$, the $x$-marginal of which will be $\pi(\cdot)$.  At each step of the Markov chain, a fresh value for $p$ is drawn from its conditional distribution given the current $x$ state, and then the relevant Hamiltonian flow is approximated for $T$ units of time to produce the next proposed move.  The resulting proposal map is
\begin{equation} \label{eqn:hmap}
g_\xi(x) = \text{Pr}_x \circ \varphi_T(x,p),
\end{equation}
where $\text{Pr}_x$ denotes the projection operator onto the $x$ coordinate, $\varphi_T$ the approximate flow for $T$ units of time, and $\xi = \{T,p\}$.  Typically the distribution for $p$ is chosen to be a $d$-dimensional Gaussian.  If the law of $p$ does not depend on $x$, then the St\"{o}rmer--Verlet (or \emph{leapfrog}) numerical integrator is typically used to approximate the flow, with $\e>0$ chosen as the integrator step-size and $L$ the number of `leapfrog steps' (meaning $T = L\e$).  The choice of $T$ is a point of ambiguity; often it is set to be some fixed value, however heuristics have also been suggested for choosing this dynamically (e.g. \citep{hoffman-gelman:2013}).  For $T = \e$ (meaning $L = 1$) in fact HMC reduces to MALA.  In general, however, for $L > 1$ (\ref{eqn:hmap}) will be a non-linear function of $p$, making analysis of the method challenging, particularly in the case of a dynamic $T$.

Our main contribution is to establish conditions under which common HMC implementations produce a geometrically ergodic Markov chain.  We also establish instances where convergence will \emph{not} be geometric, meaning the sampler may perform poorly in practice.  We first consider the case where the choice of \emph{integration time} $T$ is chosen independently of the current position, and show that here the non-linear terms in $g_\xi(x)$ can be \emph{bounded in probability} as the norm $\|x\| \to \infty$ under suitable assumptions, meaning that geometric convergence essentially occurs for HMC in the same scenarios as for MALA, when the tails of $\pi(x)$ are uniformly exponential or lighter, but no lighter than that of a Gaussian density.  We then consider an idealised scheme in which $T$ is chosen as a function of the current position, and show that in this case geometrically converging chains can be constructed for a much broader class of targets.  Although the latter results are in an idealised case, they do offer some practical guidelines for the choice of integration time, which can be used to examine some commonly used heuristics in the literature as well as suggest alternatives.

\subsection{Literature review}
The HMC method was first introduced in lattice field theory \citep{duane1987hybrid}, as a \emph{hybrid} of two differing approaches to molecular simulation introduced in \citep{alder1959studies} and \citep{metropolis1953equation} respectively.  A statistically-oriented review is given in \citep{neal2011mcmc}.  Several extensions have been suggested. A \emph{generalized} scheme in which the momentum is only partially refreshed was introduced in \citep{horowitz1991generalized} (see also \citep{ottobre2016function}).  Other extensions have been proposed to allow more directed motion and reduced rejections (e.g. \citep{campos2015extra}).  A dynamic approach to tuning the integration time parameter was introduced through the `No-U-Turn Sampler' of \citep{hoffman-gelman:2013}, which is now implemented in the Stan software \citep{carpenter2016stan}.  An extension showing how to implement the sampler on a Riemannian manifold which is \emph{globally diffeomorphic} to $\R^d$ is given in \citep{girolami2011riemann} (see also \citep{betancourt2013general}), and to embedded manifolds with closed form geodesics in \citep{byrne2013geodesic}.

Theoretical study of MCMC methods is in the main focused on two themes: \emph{convergence to equilibrium} and \emph{asymptotic variance}.  The first is often understood through upper bounding some suitable discrepancy between the $n$th iterate of the Markov chain and its limiting distribution, as a function of $n$.  When the discrepancy is taken as either the Total Variation or $V$-norm distance (for some suitable Lyapunov function $V:\X \to [1,\infty)$), then the drift and minorisation conditions popularised in \citep{meyn2012markov} can be used to show that the distance to equilibrium decreases geometrically in $n$ (we elaborate in Section \ref{sec:prelim}).  If such a bound holds then for reversible chains a Central Limit Theorem exists for long-run averages of $L^2(\pi)$ functionals (e.g. \citep{roberts2004general}).  We take this approach here.  Note that such techniques rely crucially on the chain being $\psi$-irreducible for some $\sigma$-finite measure $\psi(\cdot)$.

For HMC, \citep{cances2007theoretical} establish that if the \emph{potential energy} $U(x) = -\log\pi(x)$ is bounded above, continuous and has bounded derivative then the algorithm will produce a $\pi$-irreducible  chain.  The result holds for both the exact flow and the leapfrog integrator variants of HMC.  Typically the boundedness assumption on $U(x)$ will only be satisfied when $\X$ is compact.  The authors also show that $\pi$-irreducibility can be established more broadly if the integration time is chosen stochastically.  More recently, \citep{bou2015randomized} consider a continuous-time version of HMC in which the integration step-size is randomly sampled from an Exponential distribution.  Under the assumption that Hamilton's equations can be exactly integrated, they prove that the algorithm will produce a geometrically ergodic Markov chain whenever the tails of $\pi(x)$ decay at a Gaussian rate or faster.  The method of the authors is to relate HMC to \emph{underdamped} Langevin dynamics, the ergodic properties of which are established in \citep{mattingly2002ergodicity}.  By contrast, we relate HMC to \emph{overdamped} Langevin dynamics, as analysed in \citep{roberts1996exponential}.  Although at first this may seem less natural, in fact it allows us to paint a broad picture of when the algorithm as used in practice will and will not produce a geometrically ergodic Markov chain.  In \citep{seiler2014positive} some practical approximations are given for convergence bounds under a \emph{positive curvature} assumption on the underlying chain.  We discuss these further in Section \ref{sec:discussion}. We comment further on connections between HMC and Langevin dynamics in the supplementary material \citep{supplement}.

Asymptotic variances of long-run averages from Markov chains are often considered via analysing the expected squared jump distance $\mathbb{E}[(X_{i+1} - X_i)^2]$; at equilibrium this can then be optimised over the various parameters of the dynamics.  Careful study of this quantity can also indicate how algorithm performance depends on $d$.  In the case of HMC such analysis has been performed \citep{beskos2013optimal}, suggesting that  the method scales more favourably than other approaches with dimension, and a larger optimal acceptance rate is attained.

Recently a HMC has been generalised to the context of sampling on spaces of infinite dimension \citep{beskos2011hybrid}.  Due to the frequent singularity of measures in such spaces, it is often necessary to characterise distance to equilibrium here through other metrics than Total Variation.  Such analysis is beyond the scope of this paper, though we note that recent work in the context of MALA in \citep{eberle2014error} and \citep{durmus2015quantitative} are useful pre-cursors in this direction.


\subsection{Notation} Let $(\X,\B)$ denote a Borel space.  Here we restrict attention to $\X = \R^d$ (and write $\|x\|$ for the Euclidean norm of $x \in \X$).  For functions $f,g:\R_{\geq 0} \to \R_{\geq 0}$ let $f(x) \asymp g(x)$ mean that $\lim_{x \to \infty} f(x)/g(x) = c$ for some $c<\infty$.  Throughout let $\pi(\cdot)$ be a finite `target' measure, and $\pi(x)$ the corresponding Lebesgue density for some $x \in \X$, and let $\mathfrak{L}(\cdot)$ be a distribution defined over $\mathbb{Z}_+$.  We will denote Lebesgue measure on $\R^d$ with $\mu^L(\cdot)$, the Dirac point mass at $x$ with $\delta_x(\cdot)$ and the standard Gaussian measure with $\mu^G(\cdot)$. We write $P: \X \times \B \to [0,1]$ to denote a Markov transition kernel, meaning $P(x,\cdot)$ is a probability measure for any $x \in \X$ and $P(\cdot,A)$ is measurable for any $A \in \B$.  $P$ acts to the left on measures through $\mu P(\cdot) := \int \mu(dx)P(x,dy)$ and to the right on functions through $Pf(x) := \int f(y)P(x,dy)$.  We let $P^n(x,\cdot) := \int P^{n-1}(x,dy)P(y,\cdot)$ and say $\pi(\cdot)$ is \emph{invariant} for $P$ if $\pi P(\cdot) = \pi(\cdot)$.

Denote the Total Variation distance between two distributions $\mu(\cdot)$ and $\nu(\cdot)$ on $(\X,\B)$ as $\|\mu(\cdot) - \nu(\cdot)\|_{TV} := \sup_{|f| \leq 1}|\mathbb{E}_\mu f - \mathbb{E}_\nu f|$.  We say $\pi(\cdot)$ is a \emph{limiting} distribution for $P$ if $\|P^n(x,\cdot) - \pi(\cdot)\|_{TV} \to 0$ as $n \to \infty$, for $\pi$-a.e. $x \in \X$.  Recall that an invariant distribution $\pi(\cdot)$ will be the unique limiting measure if $P$ is both $\pi$-irreducible and aperiodic (e.g.  \citep{tierney1994markov}).  We note that the convergence results presented here could equivalently be shown under the $V$-norm distance \citep{roberts1997geometric}.

\section{Overview of Main Results}
\label{sec:main}

The majority of results in this paper concern the version of HMC which is typically used in practice, in which the `integration time' for a typical proposal is chosen independently of the current position in the chain.  In this scenario we have the following result.

\begin{thm} If Assumptions \textbf{A1} (on page \pageref{ass:A1}), \textbf{A2} (on page \pageref{ass:A2}) and \textbf{A3} (on page \pageref{ass:A3}) hold, then a Markov chain produced by the Hamiltonian Monte Carlo method (outlined in Algorithm \ref{alg:hmc}) will be geometrically ergodic.
\end{thm}

Assumption \textbf{A1} introduces a controlled degree of randomness into the integration time parameter, which ensures ergodicity of the HMC transition kernel.  Instead of establishing $\pi$-irreducibility directly on on the multiple step HMC transition, we make a simple stochasticity assumption on the integration time parameter, which allows much of the technical difficulty to be sidestepped.  Assumption \textbf{A2} imposes conditions on the distribution from which expectations are desired, essentially restricting the tail behaviour to be lighter than a Laplacian but no lighter than a Gaussian distribution.  This is to ensure that when the chain is very far from the `centre' of the space then typical proposals will bring it back to regions where probability mass concentrates.  Assumption \textbf{A3} relates to the Metropolis--Hastings acceptance rate, ensuring that this does not behave undesirably, in the sense that desirable proposals are often rejected.  We make these arguments precise in Section \ref{sec:fixed}.

We also present the following conditions under which Markov chains produced using HMC will not be geometrically ergodic.

\begin{thm} \label{thm:negative}
If either of the following hold then HMC will \emph{not} produce a geometrically ergodic Markov chain:

\vspace{0.2cm}

(i) $\lim_{\|x\| \to \infty}\frac{\|\nabla U(x)\|}{\|x\|} = \infty$ and (\ref{eqn:oscillations}) and (\ref{eqn:oscillation2}) are satisfied

\vspace{0.2cm}

(ii) There is an $M<\infty$ such that $\|\nabla U(x)\| \leq M$ for all $x \in \X$, and $\mathbb{E}_\pi[e^{s\|x\|}]=\infty$ for every $s>0$.

\end{thm}

The first of these scenarios in essence covers the case where the distribution of interest has lighter tails than those of a Gaussian distribution.  In this case explicit numerical solvers for Hamilton's equations typically become unstable in some regions of the state space.  The second is concerned with `heavy tailed' distributions, in which the resulting Hamiltonian flow can be slow, precluding a geometric rate of convergence.

To give some intuition for these results, we consider the Exponential Family class of models first introduced in \citep{roberts1996exponential} and denoted $\mathcal{E}(\beta,\alpha)$, in which $\pi(x) \in C^1(\R)$ and for all $|x|>M$ for some $M<\infty$ it holds that
\[
\pi(x) \propto \exp \left( -\alpha|x|^{\beta} \right)
\]
for some $\alpha,\beta > 0$ and any $x \in \R$.  Different choices of $\beta$ correspond to different tail behaviours, with larger values resulting in `lighter' tails.  For $\beta \geq 1$ the density is log-concave, and the specific choices $\beta = 1$ and $\beta = 2$ correspond to Laplace and Gaussian distributions.

\begin{crly} \label{crly:expfam}
For the exponential family class of models $\mathcal{E}(\beta,\alpha)$, under assumption \textbf{A1} the following results hold:

\vspace{0.2cm}

(i) For $1 \leq \beta \leq 2$, the Hamiltonian Monte Carlo method will produce a geometrically ergodic chain (for small enough $\e$ in the $\beta = 2$ case)

\vspace{0.2cm}

(ii) If $\beta < 1$ or $\beta > 2$, then the Hamiltonian Monte Carlo method will \emph{not} produce a geometrically ergodic chain
\end{crly}

\begin{proof}
See page \pageref{proof:expfam}.
\end{proof}

The results are analogous to those found for the Metropolis-adjusted Langevin algorithm in \citep{roberts1996exponential}.  A key finding of this work is that when the integration time parameter is chosen in a manner which is independent of the current position, then the two methods essentially coincide in terms of presence or absence of geometric ergodicity.  In other words, taking more than a single leapfrog step in the method will not result in a chain `becoming' geometrically ergodic, even though it may still improve the speed of convergence.

We also consider an idealised version of the method in Section \ref{sec:dynamic}, in which the integration time is allowed to depend on the current position in a prescribed way.  This scheme was designed to mimic several more recent versions of HMC (e.g. \citep{hoffman-gelman:2013}) which are commonly used in modern software packages (e.g. \citep{carpenter2016stan}). For a specific one-dimensional class of smooth exponential family models we find the following 

\begin{thm}
For the one-dimensional class of distributions with densities of the form
\[
\pi(x) \propto \exp \left( -\beta^{-1} (1+x^2)^{\beta/2} \right),
\]
then the idealised Hamiltonian Monte Carlo method introduced in Section \ref{sec:dynamic} will produce a geometrically ergodic Markov chain for any choice of $\beta>0$.
\end{thm}

The positive result in the case where $\beta>2$ is an artefact of the assumption that Hamilton's equations can be exactly solved in the idealised scheme - this result would disappear if a typical explicit numerical solver were used instead.  However, the findings for the case $\beta<1$ suggest that there are advantages to using an position-dependent integration time in the presence of heavy tails.  We discuss this in more detail in Section \ref{sec:discussion}.

\section{Preliminaries}
\label{sec:prelim}

The approach taken here to establishing geometric convergence was popularised in the monograph \citep{meyn2012markov}.  A key observation first shown in that work is the following.

\begin{thm}
Consider a $\pi$-irreducible aperiodic Markov chain with state space $(\mathbf{X}, \mathcal{B})$ and transition kernel $P$. If there exists a $\pi$-a.e. finite \emph{Lyapunov} function $V:\X \to [1,\infty]$ with `small' level sets, such that the condition $PV(x) \leq \lambda V(x) + b\mathbbm{1}_{C_\omega}(x)$ holds for some $\lambda < 1$, $b<\infty$ and some set $C_\omega:=\{x : V(x)\leq \omega\}$ with $\omega<\infty$, then $\exists \rho<1$ and a $\pi$-a.e. finite $M:\X\to[0,\infty]$ such that
\begin{equation} \label{eqn:geoerg}
\|P^n(x,\cdot) - \pi(\cdot)\|_{TV} \leq M(x)\rho^n.
\end{equation}
\end{thm}
Recall that a set $C \in \mathcal{B}$ is called `small' if there is a $t<\infty$, a measure $\nu(\cdot)$ defined on $(\mathbf{X},\mathcal{B})$ and an $\epsilon>0$ such that $\forall x \in C$ and $\forall A \in \mathcal{B}$ it holds that $P^t(x,A) \geq \epsilon \nu(A)$ (see e.g. \citep{roberts2004general}).

We are concerned here with specific forms of $P$.

\begin{defn} We say $P$ is of the \emph{Metropolis--Hastings} type if
\begin{equation} \label{eqn:mhtype}
P(x,dy) = \alpha(x,y)Q(x,dy) + r(x)\delta_x(dy),
\end{equation}
where $Q$ is a Markov kernel, $\alpha(x,y)$ is defined in (\ref{eqn:accept}) and $r(x) = 1- \int \alpha(x,y)Q(x,dy)$.
\end{defn}

The following was shown in \citep{roberts1996geometric} when $P$ is of the form (\ref{eqn:mhtype}).

\begin{prop} \label{prop:mh1}
If $\pi(\cdot)$ and $Q(x,\cdot)$ admit Lebesgue densities $\pi(x)$ and $q(y|x)$, $\pi(x)$ is bounded away from $0$ and $\infty$ on compact sets, and there exists $\delta_q>0$ and $\epsilon_q>0$ such that, for every $x$,
\[
\|x-y\| \leq \delta_q \implies q(y|x) \geq \epsilon_q,
\]
then the Metropolis--Hastings chain with candidate density $q(y|x)$ is $\pi$-irreducible and aperiodic, and all compact sets are small.
\end{prop}

\begin{crly} If $P$ is of Metropolis--Hastings type and the conditions of Proposition \ref{prop:mh1} are satisfied, then (\ref{eqn:geoerg}) holds if and only if
\begin{equation} \label{eqn:drift}
\limsup_{\|x\|\to\infty}\frac{PV(x)}{V(x)} < 1,
\end{equation}
for some Lyapunov function $V$.
\end{crly}

Showing a lack of geometric ergodicity typically requires careful study of the distribution of return times to small sets.  The following result of \citep{roberts1996geometric}, however, provides a straightforward method for doing this for Metropolis--Hastings kernels.

\begin{prop} \label{prop:supr}
If $P$ is of Metropolis--Hastings type, then (\ref{eqn:geoerg}) \emph{fails} to hold if
$\esssup r(x) = 1$.
\end{prop}

Lack of geometric ergodicity can also be established in some cases using the following result of \citep{jarner2003necessary}.

\begin{prop} \label{prop:tight}
If for any $\eta > 0$ there is a $\delta > 0$ such that
\begin{equation} \label{eqn:lackge2}
P(x,B_\delta(x)) > 1-\eta,
\end{equation}
where $B_\delta(x) :=\{ y \in \X : \|x-y\| < \delta \}$, then $P$ can be geometrically ergodic \emph{only} if $\mathbb{E}_\pi [e^{\beta \|x\|}] < \infty$ for some $\beta > 0$.
\end{prop}

If $P$ is of Metropolis--Hastings type, it is straightforward to verify that $Q(x,B_\delta(x)) > 1- \eta$ ensures (\ref{eqn:lackge2}), meaning we only need consider the candidate kernel in these cases.

\section{Hamiltonian Monte Carlo}
\label{sec:hmc}

We give a brief introduction here.  For a more detailed account see \citep{neal2011mcmc} or \citep{betancourt2014geometric}.  We consider probability densities of the form $\pi(x) \propto e^{-U(x)}$ for some $U:\X \to [0,\infty)$.  If we view $U(x)= -\log\pi(x)$ as a `potential' energy in a physical system, it is natural to consider the larger phase space and construct the \emph{Hamiltonian}
\begin{equation} \label{eqn:hamiltonian}
H(x,p) = U(x) + \frac{1}{2}p^t M^{-1} p,
\end{equation}
where $p$ denotes a $d$-dimensional `momentum' variable, $M$ a $d \times d$ `mass' matrix and $p^t M^{-1}p/2$ the `kinetic' energy (other forms of kinetic energy are also possible, see e.g. \citep{girolami2011riemann}).  Provided $U(x)$ is differentiable, we can evolve the coordinates $(x_t,p_t)$ through time in such a way that $H(x_t,p_t) = H(x_{t+s},p_{t+s})$ for any $t,s \in \mathbb{R}$ using Hamilton's equations
\begin{equation} \label{eqn:hamilton}
\frac{d p_t}{dt} = -\frac{\partial H}{\partial x}, ~~ \frac{dx_t}{dt} = \frac{\partial H}{\partial p}.
\end{equation}
Solving (\ref{eqn:hamilton}) results in \emph{Hamiltonian flow}.  To put this presentation into the framework introduced in Section \ref{sec:intro}, we can consider constructing a measure-preserving map $f_\theta:\X \to \X$ by setting the input to be $x_0$, choosing a momentum variable $p_0$, solving (\ref{eqn:hamilton}) for $T$ units of time and then projecting back down onto $\X$ to produce $x_T$.  The parameters $\theta = \{p_0,T\}$ define the behaviour of a single map $f_\theta$, and how they are chosen define the behaviour of the Markov chain produced by iterating the process of randomly selecting a $\theta$ and then applying the resulting map $f_\theta$ to the current point to produce the next.

Of course, it is often not possible to solve (\ref{eqn:hamilton}) exactly, so numerical methods are needed.  Fortunately, the rich geometric structure of Hamiltonian systems allows the construction of \emph{symplectic} integrators, which possess attractive long term numerical stability properties (e.g. \citep{leimkuhler2004simulating}), meaning that for appropriate Hamiltonians the approximate solution of (\ref{eqn:hamilton}) is such that $H(x_t,p_t) \approx H(x_0,p_0)$ for all $t<\eta$, where $\eta \gg 0$.  The standard choice when the Hamiltonian is of the form (\ref{eqn:hamiltonian}) is the St\"{o}rmer--Verlet or \emph{leapfrog} scheme, in which $(x_{L\e},p_{L\e})$ is generated from $(x_0,p_0)$ using $L$ steps of the recursion
\begin{align*}
p_{t+\frac{\e}{2}} &= p_t - \frac{\e}{2}\nabla U(x_t), \\
x_{t+\e} &= x_t + \e M^{-1}p_{t+\frac{\e}{2}}, \\
p_{t+\e} &= p_{t+\frac{\e}{2}} - \frac{\e}{2}\nabla U(x_{t+\e}),
\end{align*}
for some step-size $\e>0$.  Although the resulting \emph{approximate} flow map $\varphi_{L\e}(x_0,p_0) := (x_{L\e},p_{L\e})$ no longer preserves $\pi(\cdot)$, it can be used as a proposal mechanism within the Metropolis--Hastings framework (e.g. \citep{neal2011mcmc}).  The full method is shown in Algorithm \ref{alg:hmc} below.

\begin{algorithm}[ht] 
\caption{Hamiltonian Monte Carlo, single iteration.}
\begin{algorithmic}
\label{alg:hmc}
\REQUIRE $x_{i-1}$, $\epsilon \geq 0$, $\mathfrak{L}(\cdot)$
\STATE Set $x_0 \leftarrow x_{i-1}$, draw $p_0 \sim N(0,M)$, $L \sim \mathfrak{L}(\cdot)$, set $T \leftarrow L\varepsilon$
\STATE Draw $u \sim U[0,1]$
\STATE Set $\delta \leftarrow H(x_0,p_0) - H\circ \varphi_T(x_0,p_0)$,
\IF{$\log(u) < \delta$}
\STATE Set $x_i \leftarrow \text{Pr}_x \circ \varphi_T(x_0,p_0)$
\ELSE
\STATE Set $x_i \leftarrow x_{i-1}$
\ENDIF
\end{algorithmic}
\end{algorithm}

\begin{rem}
From this point forward we assume $M = I$ for ease of exposition but without loss of generality.
\end{rem}

\subsection{The marginal chain} 
To use the techniques of \citep{meyn2012markov}, it is helpful to express the HMC transition in such a way that when $\|x\|$ is large it is clear how the chain will behave.  Although it is typically presented as a map on the larger phase space, HMC can simply be thought of as a Markov chain on $\X$, and we will find this representation useful in relation to the above.  In this case the candidate map $g_\xi$ is given by the following proposition, which can be straightforwardly be derived using classical results (see e.g. \citep{bou2018geometric}).

\begin{prop} \label{prop:mprop}
The HMC candidate map can be written
\begin{equation} \label{eqn:mprop}
x_{L\e} = x_0 - \frac{L\e^2}{2}\nabla U(x_0) - \e^2\sum_{i=1}^{L-1} (L-i)\nabla U(x_{i\e}) + L\e p_0.
\end{equation}
where $p_0 \sim N(0,I)$, $L$ is the number of leapfrog steps and $\e$ the integrator step-size.  With this choice, the acceptance probability will be
\begin{equation}
\alpha(x_0,x_{L\e}) = 1 \wedge \frac{\pi(x_{L\e})}{\pi(x_0)}\exp \left( \frac{1}{2}\|p_0\|^2 - \frac{1}{2}\|p_{L\e}\|^2 \right),
\end{equation}
where
\begin{equation} \label{eqn:mprop2}
p_{L\e} = p_0 - \frac{\e}{2}\nabla U(x_0) - \e\sum_{i=1}^{L-1} \nabla U(x_{i\e}) - \frac{\e}{2}\nabla U(x_{L\e}).
\end{equation}
\end{prop}

Proposition \ref{prop:mprop} highlights the previously noted relationship between HMC and MALA quite explicitly, as setting $L = 1$ means the third term on the right-hand side of (\ref{eqn:mprop}) disappears, leaving the MALA proposal $x_0 - \e^2\nabla U(x_0)/2 + \e p_0$.  It also highlights why taking $L >1$ proposes a greater challenge, as for each $x_{i\e}$ with $i \geq 1$ this term will typically be a nonlinear transformation of $x_0$ and $p_0$.  As $p_0$ is stochastic, then $\e^2\sum_{i=1}^{L-1} (L-i)\nabla U(x_{i\e})$ will be also, but its distribution will often be intractable.

\section{Results for an position-independent integration time}
\label{sec:fixed}

In this section we make the assumption that the distribution $\mathfrak{L}(\cdot)$ for the number of leapfrog steps $L$ does not depend on the current position.   This is relaxed in Section \ref{sec:dynamic}.

\subsection{$\pi$-irreducibility}

It is known (e.g. \citep{cances2007theoretical}) that establishing $\pi$-irreducibility is not so straightforward in the case of HMC as for Metropolis--Hastings methods based on random walks or Langevin diffusions.  The canonical example where the system becomes reducible is integrating the harmonic oscillator over precisely one period (e.g. \citep{leimkuhler2004simulating}).  We show this in the supplementary material \citep{supplement}.


The observation noted here and elsewhere that HMC in the case $L=1$ corresponds to MALA, for which irreducibility is established in \citep{roberts1996exponential}, can be exploited to alleviate these issues and establish $\pi$-irreducibility of HMC under the following assumption.

\vspace{0.35cm}

\noindent \textbf{A1} {\itshape The distribution $\mathfrak{L}(\cdot)$ is such that $\mathbb{P}_{\mathfrak{L}}[L=1]>0$, and for any fixed $(x_0,p_0)\in \R^{2d}$ and $\e>0$, and that there is an $s <\infty$ such that $\mathbb{E}_\mathfrak{L} [e^{s\|x_{L\e}\|}]<\infty$.} \label{ass:A1}

\vspace{0.35cm}

When assumption \textbf{A1} holds then the fact that HMC produces an ergodic Markov chain can be straightforwardly invoked from existing MALA results \citep{roberts1996exponential}.  The idea of randomising the integration time is commonly recommended for practical applications of the method (e.g. \citep{neal2011mcmc,girolami2011riemann}), and more theoretical motivation for doing so is given in \citep{betancourt2016identifying}.  The finite exponential expectation condition is needed to ensure that the Lyapunov function used to prove geometric ergodicity is valid.  One simple way to ensure this in practice (under the additional assumptions imposed on the potential $U$ in the next subsection) is that $\mathbb{P}_{\mathfrak{L}}[L>l]=0$ for some fixed $l<\infty$, though weaker conditions than this are also possible.

\begin{rem} 
Assumption \textbf{A1} can be viewed as the discrete time analogue to the the exponential integration time assumption made in \citep{bou2015randomized}, and in many respects serves a similar purpose.  Similar conditions are also exploited to prove $\pi-$irreducibility results in \citep{cances2007theoretical}.
\end{rem}

\begin{rem}
In fact, before the final publication of the present work, it was shown in \citep{durmus2017convergence} that $\pi$-irreducibility can indeed be established without using assumption \textbf{A1}, but instead considering the HMC chain using a fixed number of leapfrog steps $L\geq 1$, under suitable assumptions and using appropriate techniques. We refer the interested reader to that work for details.
\end{rem}

\subsection{Geometric ergodicity}  

We first present here some seemingly abstract conditions under which the HMC method produces a geometrically ergodic Markov chain.  We then give some natural assumptions on the potential $U(x)$ under which these hold.

We present the results of this section conditioned on a fixed choice of the number of leapfrog steps $L$, for ease of exposition.  Note that the required drift conditions shown hold for a fixed $L$, then under \textbf{A1} they will hold when possible values for $L$ are averaged over according to $\mathfrak{L}(\cdot)$, so this does not affect the generality of the results.

\textit{Notation.} We introduce some further notation for this section. Let $I_\delta(x) := \{ y \in \X : \|y\| \leq \|x\|^\delta \}$ for some $1/2 < \delta < 1$. In the case $\delta=1$ we will simply write $I(x)$.  Let
\[
m_{L,\e}(x_0,p_0) := x_0 - L\e^2\nabla U(x_0)/2 - \e^2\sum_{i=1}^{L-1} (L-i)\nabla U(x_{i\e})
\]
denote the `mean' next candidate position ($x_{L\e}-L\e p_0$), and 
\[
\psi_{L,\e}(x_0,p_0):= L\e^2\nabla U(x_0)/2 + \e^2\sum_{i=1}^{L-1} (L-i)\nabla U(x_{i\e})
\]
denote the proposal `drift' (implying $m_{L,\e}(x_0,p_0) = x_0 - \psi_{L,\e}(x_0,p_0)$).  We will also sometimes write $h:=\e^{2}/2$
in a most likely futile attempt to keep things readable.

\begin{thm} \label{thm:geoerg} The HMC method produces a geometrically ergodic Markov chain if assumption \textbf{A1} holds, and in addition both
\begin{equation} \label{eqn:mean}
\limsup_{\|x_0\|\to\infty, \|p_0\| \leq \|x_0\|^\delta} \left( \|m_{L,\e}(x_0,p_0)\| - \|x_0\| \right) < -\sqrt{2}L\e\eta(d)
\end{equation}
where $\eta(d):=\Gamma((d+1)/2)/\Gamma(d/2)$, and
\begin{equation} \label{eqn:inwards}
\lim_{\|x_0\|\to\infty} \int_{R(x_0) \cap I(x_0)} Q(x_0,dy) = 0,
\end{equation}
where $R(x_0) := \{ y \in \X : \alpha(x_0,y) < 1 \}$ denotes the `potential rejection region' and $I(x_0) := \{ y \in \X : \|y\| \leq \|x_0\| \}$ the `interior' of $x_0$.
\end{thm}

\begin{proof}
Take $V(x) = e^{s\|x\|}$ for some $s > 0$ and write $A(x_0) := R(x_0)^c$. Then we can write
\begin{align*}
\frac{PV(x_0)}{V(x_0)} 
&= 
\int_{A(x_0)} e^{s(\|y\| - \|x_0\|)}Q(x_0,dy) + \int_{R(x_0)} e^{s(\|y\| - \|x_0\|)}\alpha(x_0,y)Q(x_0,dy) \\
&\qquad + \int_{R(x_0)} (1-\alpha(x_0,y))Q(x_0,dy) 
\\
&= 
\int_{\mathbb{R}^d} e^{s(\|y\| - \|x_0\|)}Q(x_0,dy) + \int_{R(x_0)} \left(1 - e^{s(\|y\| - \|x_0\|)} \right) (1-\alpha(x_0,y)) Q(x_0,dy) 
\\
&\leq
\int e^{s(\|y\| - \|x_0\|)}Q(x_0,dy) + \int_{R(x_0) \cap I(x_0)} Q(x_0,dy).
\end{align*}
The last integral asymptotes to zero as $\|x_0\| \to \infty$ by (\ref{eqn:inwards}).  Writing $x_{L\e}(p_0)$ to indicate that $x_{L\e}$ depends on $p_0$, the first integral can be written
\begin{equation} \label{eqn:brokenupintegral}
\int_{I_\delta(x_0)} e^{s(\|x_{L\e}(p_0)\| - \|x_0\|)} \mu^G(dp_0) + \int_{I_\delta(x_0)^c} e^{s(\|x_{L\e}(p_0)\| - \|x_0\|)} \mu^G(dp_0).
\end{equation}
Noting that $\|x_{L\e}(p_0)\| \leq \|m_{L,\e}(x_0,p_0)\| + L\e\|p_0\|$ for large enough $\|x_0\|$ and using (\ref{eqn:mean}) above then setting $\xi(x_0):=\sup_{\|p_0\|\leq \|x_0\|^\delta}(\|m_{L,\e}(x_0,p_0)\|-\|x_0\|)$ we can write
\[
\int_{I_\delta(x_0)} e^{s(\|x_{L\e}\| - \|x_0\|)} \mu^G(dp_0) \leq e^{s\xi(x_0)} \int_{I_\delta(x_0)} e^{sL\e\|p_0\|}\mu^G(dp_0).
\]
The last integral can be bounded above by the moment generating function of a Chi-distributed random variable with $d$ degrees of freedom, and so equals $e^{\log(1+s\sqrt{2}L\e \eta(d) +o(s))} \leq e^{s\sqrt{2}L\e \eta(d) +o(s)}$.  Therefore by (\ref{eqn:mean}) the integral asymptotes to a quantity which is strictly less than one if $s>0$ is chosen to be suitably small.

It remains to show that the right-hand integral in (\ref{eqn:brokenupintegral}) becomes negligibly small as $\|x_0\| \to \infty$.  It follows from (\ref{eqn:mprop}) and (A4.4.1) that $\|x_{L\e}\| \in O(\max(\|x_0\|,\|p_0\|))$.  This means that for some constants $C \in \R$ and for $p_0 \in I_\delta(x_0)^c$ and $\|x_0\|$ large enough we can write
\begin{align*}
\exp\left( s\|x_{L\e}\| - s\|x_0\| - \frac{1}{2}\|p_0\|^2  \right) &\leq \exp\left( c\max(\|x_0\|,\|p_0\|) - \frac{1}{2}\|p_0\|^2 \right), \\
&= \exp\left( \|p_0\|\left( C\frac{\max(\|x_0\|,\|p_0\|)}{\|p_0\|} - \|p_0\|\right) \right), \\
&\leq \exp\left( \|p_0\|\left( C\|x_0\|^{1 - \delta} - \|p_0\|\right) \right)
\end{align*}
Provided $\delta > 1/2$, then for $\|x_0\|$ large enough $C\|x_0\|^{1-\delta} - \|x_0\|^\delta < -1$, meaning
\[
\exp\left( \|p_0\|\left( C\|x_0\|^{1 - \delta} - \|p_0\|\right) \right) \leq \exp\left(-\|p_0\|\right),
\]
meaning
\[
\int_{I_\delta(x_0)^c} e^{s(\|x_{L\e}\| - \|x_0\|)}\mu^G(dp_0) \leq \int_{I_\delta(x_0)^c} e^{-\|p_0\|} \mu^G(dp_0) \leq 2e^{-\|x_0\|^\delta},
\]
which becomes negligibly small as $\|x_0\| \to \infty$, as required.
\end{proof}

Theorem \ref{thm:geoerg} is a generalisation of Theorem 4.1 in \citep{roberts1996exponential} to the HMC case.  The nontriviality involved in this extension is accounting for the randomness induced into $m_{L,\e}(x_0,p_0)$ from $p_0$.

\subsubsection{Requirements for (\ref{eqn:mean}) to be satisfied.}

In the case $L=1$ (\ref{eqn:mean}) corresponds to
\begin{equation} \label{eqn:mean1step}
\|x_{0}-h\nabla U(x_{0})\|-\|x_{0}\|< -\sqrt{2}\e\eta(d)
\end{equation}
whenever $\|x_{0}\|>M$ for some $M<\infty$.  The statements in this section give three simple conditions which establish this are \emph{also} sufficient to establish (\ref{eqn:mean}) when $L\geq2$. The main result is stated below.
The crucial consequence of this is that controlling the behaviour
of `global move' updates produced by HMC when $L>1$ can be done through
only `local' knowledge, meaning analytic information at the current
point $x_{0}$.
\begin{thm}
For any $L\geq1$ (\ref{eqn:mean}) holds if the following conditions are met
\vspace{0.3cm}

(SC1.1) $\lim_{\|x_{0}\|\to\infty}\|\nabla U(x_{0})\|=\infty$
\vspace{0.3cm}

(SC1.2) $\liminf_{\|x_{0}\|\to\infty}\frac{\langle\nabla U(x_{0}),x_{0}\rangle}{\|\nabla U(x_0)\|\|x_{0}\|} >0$
\vspace{0.3cm}

(SC1.3) $\lim_{\|x_{0}\|\to\infty}\frac{\|\nabla U(x_{0})\|}{\|x_{0}\|}=0$.
\vspace{0.3cm}

\noindent In addition, if (SC1.3) is replaced by
\vspace{0.3cm}

(SC1.3b) $\limsup_{\|x_{0}\|\to\infty}\frac{\|\nabla U(x_{0})\|}{\|x_{0}\|}=S_{l},$
\vspace{0.3cm}

\noindent for some $S_{l}<\infty$, then there is an $\e_{0}\in(0,\infty)$
such that the same result holds provided $\e\in(0,\e_{0})$.
\end{thm}
\begin{proof}
This is re-stated as Proposition \ref{prop:SC1.3} and Proposition
\ref{prop:SC1.3b} below, which follow from the preceding Lemmas.
\end{proof}
The conditions of the result are intuitive. Condition (SC1.2) ensures
that the gradient asymptotically `points inwards', while (SC1.1) and (SC1.3) ensure that  $\|\nabla U(x_{0})\|$ grows but at an asymptotically sublinear rate. We begin with a straightforward observation.

\begin{prop}
Sufficient conditions such that
\[
\limsup_{\|x_{0}\|\to\infty}\left(\|x_{0}-h\nabla U(x_{0})\|-\|x_{0}\|\right)<0
\]
are:


(SC1.2a) $\limsup_{\|x_{0}\|\to\infty}\left(\frac{h^{2}\|\nabla U(x_{0})\|^{2}-2h\langle\nabla U(x_{0}),x_{0}\rangle}{2\|x_{0}\|}\right)< -\sqrt{2} \e \eta(d)$
\vspace{0.3cm}

(SC1.3) $\lim_{\|x_{0}\|\to\infty}\frac{\|\nabla U(x_{0})\|}{\|x_{0}\|}=0$.
\end{prop}
\begin{proof}
First note that $\|x_{0}-h\nabla U(x_{0})\|=\sqrt{\|x_{0}\|^{2}+h^{2}\|\nabla U(x_{0})\|^{2}-2h\langle\nabla U(x_{0}),x_{0}\rangle}$.
Recall the generalised Bernoulli inequality: if $y>-1$ and $r\in[0,1]$
then $\left(1+y\right)^{r}\leq1+ry$. Setting $r:=1/2$, $a(x_{0}):=\|x_{0}\|^{2}$
and $b(x_{0}):=h^{2}\|\nabla U(x_{0})\|^{2}-2h\langle\nabla U(x_{0}),x_{0}\rangle$
then we have
\begin{align*}
a(x_{0})^{r}\left(1+\frac{b(x_{0})}{a(x_{0})}\right)^{r} &\leq a(x_{0})^{r}\left(1+\frac{b(x_{0})}{2a(x_{0})}\right) \\ 
&=\|x_{0}\|+\frac{h^{2}\|\nabla U(x_{0})\|^{2}-2h\langle\nabla U(x_{0}),x_{0}\rangle}{2\|x_{0}\|}.
\end{align*}
This will be strictly less than $\|x_{0}\|$ in the limit as $\|x_{0}\|\to\infty$
provided $b(x_{0})/a(x_{0})>-1$. Noting that
\[
\frac{b(x_{0})}{a(x_{0})}=\frac{h^{2}\|\nabla U(x_{0})\|^{2}-2h\langle\nabla U(x_{0}),x_{0}\rangle}{\|x_{0}\|^{2}}>\frac{h^{2}\|\nabla U(x_{0})\|^{2}-2h\|\nabla U(x_{0})\|\|x_{0}\|}{\|x_{0}\|^{2}},
\]
then it suffices to see that under (SC1.3) the right-hand side can
be made arbitrarily close to zero by taking $\|x_{0}\|$ large enough.
\end{proof}
We can also recover some more intuitive sufficient conditions.
\begin{crly}
\label{crly:intuitive}A more intuitive condition which implies (SC1.2a)
conditional on (SC1.3) and (SC1.1) is
\vspace{0.3cm}

(SC1.2) $\lim_{\|x_{0}\|\to\infty}\frac{\langle\nabla U(x_{0}),x_{0}\rangle}{\|\nabla U(x)\|\|x_{0}\|} >0$
\end{crly}
\begin{proof}
Using (SC1.1) and (SC1.2) gives
\[
\lim_{\|x_{0}\|\to\infty}\frac{\langle\nabla U(x_{0}),x_{0}\rangle}{\|x_{0}\|} = \infty
\]
For large enough $\|x_0\|$ this implies
\begin{align*}
\left(\frac{h^{2}\|\nabla U(x_{0})\|}{2\|x_{0}\|}-\frac{h\langle\nabla U(x_{0}),x_{0}\rangle}{2\|\nabla U(x_{0})\|\|x_{0}\|}\right) &< -\frac{\sqrt{2}\e\eta(d)}{\|\nabla U(x_0) \|} \\
\implies \limsup_{\|x_{0}\|\to\infty}\left(\frac{h^{2}\|\nabla U(x_{0})\|^{2}-2h\langle\nabla U(x_{0}),x_{0}\rangle}{2\|x_{0}\|}\right) &< -\sqrt{2}\e \eta(d).
\end{align*}
\end{proof}
From now on we refer to (SC1.1), (SC1.2) and (SC1.3) combined as (SC1.1)-(SC1.3).
Next we show that these same conditions are sufficient
for (\ref{eqn:mean}) to hold.
\begin{lem}
\label{lem:sufficient_L>1}Under the following conditions (\ref{eqn:mean}) holds:
\vspace{0.3cm}

(i) $\liminf_{\|x_{0}\|\to\infty,\|p_{0}\|\leq\|x_{0}\|^{\delta}}\left(\frac{\|\psi_{L,\e}\|^{2}-2\langle\psi_{L,\e},x_{0}\rangle}{\|x_{0}\|^{2}}\right)>-1$
\vspace{0.3cm}

(ii) $\limsup_{\|x_{0}\|\to\infty,\|p_{0}\|\leq\|x_{0}\|^{\delta}}\left(\frac{\|\psi_{L,\e}\|^{2}-2\langle\psi_{L,\e},x_{0}\rangle}{2\|x_{0}\|}\right)<0$.
\end{lem}
\begin{proof}
Using the generalised Bernoulli inequality as above gives the result.
\end{proof}
Next we relate the conditions of Lemma \ref{lem:sufficient_L>1} to
criteria that only depend on the current
point $x_{0}$. The following lemmas give a starting point.
\begin{lem}
\label{lem:U_lem_1}Provided $\|p_{0}\|\leq\|x_{0}\|^{\delta}$ and
(SC1.3) holds then we have the following
\vspace{0.2cm}

(i) For any $\eta>0$ there is an $M_{\eta}<\infty$ such that whenever
$\|x_{0}\|>M_{\eta}$ it holds that $(1-\eta)\|x_{0}\|\leq\|x_{\e}\|\leq(1+\eta)\|x_{0}\|$
\vspace{0.2cm}

(ii) $\|\nabla U(x_{\e})\|=o(\|x_{0}\|)$
\vspace{0.2cm}

(iii) $\|p_{\e}\|\in o(\|x_{0}\|)$.
\end{lem}
\begin{proof}
(i) Noting that $\|x_{\e}\|=\|x_{0}-h\nabla U(x_{0})+\e p_{0}\|$
gives 
\begin{align*}
(1-\delta)\|x_{0}\| &\leq\|x_{0}\|-h\|\nabla U(x_{0})\|-\e\|x_{0}\|^{\delta} \\
&\leq\|x_{0}-h\nabla U(x_{0})+\e p_{0}\| \\
&\leq\|x_{0}\|+h\|\nabla U(x_{0})\|+\e\|x_{0}\|^{\delta} \\
&\leq(1+\delta)\|x_{0}\|.
\end{align*}

(ii) We have from (i) and (SC1.3) that for any $\gamma>0$ there is
an $M_{\gamma}<\infty$ such that whenever $\|x_{0}\|>M_{\gamma}/(1-\delta$)
then $\|\nabla U(x_{\e})\|/\|x_{\e}\|<\gamma$. This implies
using (i) that $\|\nabla U(x_{\e})\|/\|x_{0}\|<\gamma(1-\delta)$,
and since $\gamma(1-\delta)$ can be made arbitrarily small then the
result follows.

(iii) $\|p_{\e}\|=\|p_{0}-\e\nabla U(x_{0})/2-\e\nabla U(x_{\e})/2\|\leq\|p_{0}\|+\e\|\nabla U(x_{0})\|/2+\e\|\nabla U(x_{\e})\|/2$,
which is $\in o(\|x_{0}\|)$ using (i) and (ii) and the fact that
$\|p_{0}\|\leq\|x_{0}\|^{\delta}$
\end{proof}
\begin{lem}
\label{lem:orders_L}Provided $\|p_{0}\|\leq\|x_{0}\|^{\delta}$ and
(SC1.3) holds then for any $L<\infty$ and each $i\in\{0,...,L-1\}$
the following hold
\vspace{0.2cm}

(i) For any $\eta>0$ there is an $M_{\eta}<\infty$ such that whenever
$\|x_{0}\|>M_{\eta}$ it holds that $(1-\eta)\|x_{0}\|\leq\|x_{i\e}\|\leq(1+\eta)\|x_{0}\|$
\vspace{0.2cm}

(ii) $\|\nabla U(x_{i\e})\|\in o(\|x_{0}\|)$
\vspace{0.2cm}

(iii) $\|p_{i\e}\|\in o(\|x_{0}\|)$
\vspace{0.1cm}

(iv) $\|\psi_{L,\e}\|\in o(\|x_{0}\|)$.
\end{lem}
\begin{proof}
The results follow iteratively for each $i$ using the same approach
as in the previous Lemma. For the case $i=2$ then noting that $\|x_{2\e}\|=\|x_{\e}-h\nabla U(x_{\e})+\e p_{\e}\|$,
then (i) in this case follows from Lemma \ref{lem:U_lem_1}. It follows
that $\|\nabla U(x_{2\e})\|\in o(\|x_{0}\|)$ and $\|p_{2\e}\|\in o(\|x_{0}\|)$
by an analogous argument to this Lemma. Given this then it can be
shown that (i) holds for $i=3$, and then (ii) and (iii) by the same
logic, and the argument can be iterated as many times as is needed.
The last claim follows trivially from the second.
\end{proof}
\begin{prop}
\label{prop:SC1.3}Under (SC1.1)-(SC1.3) then for any $L<\infty$
the conditions of Lemma \ref{lem:sufficient_L>1} are satisfied.
\end{prop}
\begin{proof}
First we show (i). Writing $x^{*}:=\arg\max_{i\in{0,...,L-1}}\left\{ \|\nabla U(x_{i\e})\|\right\} $,
then we have $\|\psi_{L,\e}\|\leq L\e^{2}\sum_{i=0}^{L-1}\|\nabla U(x_{i\e})\|\leq L^{2}\e^{2}\|\nabla U(x^{*})\|$,
which implies
\[
\frac{\|\psi_{L,\e}\|^{2}}{\|x_{0}\|^{2}}\leq\frac{L^{4}\e^{4}\|\nabla U(x^{*})\|^{2}}{(1-\eta)^{2}\|x^{*}\|^{2}}
\]
which can be made arbitrarily small by taking $\|x_{0}\|$ large enough
using (SC1.3). Noting that $\|\psi_{L,\e}\|/\|x_{0}\|\geq\langle\psi_{L,\e},x_{0}\rangle/\|x_{0}\|^{2}\geq-\|\psi_{L}\|/\|x_{0}\|$,
then an analogous argument can be used to show that $-2\langle\psi_{L},x_{0}\rangle/\|x_{0}\|^{2}$
will also tend to zero as $\|x_{0}\|\to\infty$. 

(ii) First note from above that $\lim_{\|x_{0}\|\to\infty}\|\psi_{L,\e}\|^{2}/\left(L^{2}\e^{2}\|\nabla U(x^{*})\|\|x_{0}\|\right)=0$.
By an analogous argument to that used in the proof of Corollary \ref{crly:intuitive},
it is clear therefore that (ii) holds if
\[
\limsup_{\|x_{0}\|\to\infty}\left(\frac{\|\psi_{L,\e}\|^{2}}{L^{2}\e^{2}\|\nabla U(x^{*})\|\|x_{0}\|}-\frac{\langle\psi_{L,\e},x_{0}\rangle}{L^{2}\e^{2}\|\nabla U(x^{*})\|\|x_{0}\|}\right)< -\sqrt{2}L\e \eta(d)
\]
which in turn holds if the statement
\[
\lim_{\|x_{0}\|\to\infty}\frac{\langle\psi_{L,\e},x_{0}\rangle}{\|x_{0}\|} = \infty
\]
does. The numerator can be decomposed as
\begin{align*}
\langle\psi_{L,\e},x_{0}\rangle &\geq \sum_{i=0}^{L-1}c_{i}\langle\nabla U(x_{i\e}),x_{0}\rangle \\
&=\sum_{i=0}^{L-1}c_{i}\langle\nabla U(x_{i\e}),x_{i\e}\rangle+\sum_{i=0}^{L-1}c_{i}\langle\nabla U(x_{i\e}),x_{0}-x_{i\e}\rangle,
\end{align*}
where each $c_{i}=(L-i)\e^{2}$ for $i\geq1$ and $c_{0}=L\e^{2}/2$
. The second of these terms is $o(\|\nabla U(x^*)\|\|x_{0}\|)$ using the Cauchy--Schwartz
inequality and Lemma \ref{lem:orders_L} (which shows that $\|x_{0}-x_{i\e}\|\in o(\|x_{0}\|)$), and so this term
vanishes if divided by $\|\nabla U(x^{*})\|\|x_{0}\|$. The first term divided by the same quantity will be strictly positive
as each term in the sum is $\geq0$ using (SC1.3) and Lemma \ref{lem:orders_L},
and at least one of them is $>0$ since it will correspond to $x^{*}$. Using (SC1.1) establishes that $\langle \psi_{L\e},x_0\rangle / \|x_0\| \to \infty$ as $\|x_0\|\to\infty$, proving the result.
\end{proof}
The condition (SC1.3) allows clarity in the proofs, but precludes
the natural boundary case of distributions with Gaussian tails. The following
proposition addresses this.
\begin{prop}
\label{prop:SC1.3b} If  (SC1.1)-(SC1.2) hold and in addition
\vspace{0.3cm}

(SC1.3b) $\limsup_{\|x_{0}\|\to\infty}\frac{\|\nabla U(x_{0})\|}{\|x_{0}\|}=S_{l}$
\vspace{0.3cm}

\noindent for some constant $S_{l}<\infty$, then there is an $\e_{0}\in(0,\infty)$
such that for any choice of $\e\in(0,\e_{0})$ the conditions
of Lemma \ref{lem:sufficient_L>1} are satisfied.
\end{prop}
\begin{proof}
We simply note that the term $\langle\psi_{L,\e},x_{0}\rangle\in O(\e^{2})$,
while $\|\psi_{L,\e}\|^{2}\in O(\e^{4})$, so that the
proofs of the preceding Lemmas can be straightforwardly modified when
(SC1.3) is replaced by (SC1.3b) by choosing a small enough value of
$\e$ that the inner product dominates the square norm. We omit
the details of this.
\end{proof}

The sensitivity to the choice of $\e$ in this case is well known in this scenario as a potential source of numerical instabilities, and choosing $\e < 1/S_l$ is recommended to alleviate such issues (e.g. \citep{leimkuhler2004simulating}).  We conclude this subsection with the following assumption that we require for a geometrically ergodic Markov chain produced by the HMC method, which is a natural conclusion of the preceding results.

\vspace{0.35cm}

\noindent \textbf{A2} {\itshape The potential $U(x)$ satisfies either (SC1.1)-(SC1.3), or it satisfies (SC1.1)-(SC1.3b) and $\e$ is chosen to be suitably small that the conditions of Lemma \ref{lem:sufficient_L>1} are satisfied.} \label{ass:A2}


\begin{rem}
Condition (SC1.1) precludes densities for which $\|\nabla U(x)\| \to c$ for some $0<c<\infty$.  Often geometric ergodicity will still hold in this case, as we demonstrate in Corollary \ref{crly:expfam}, however a different argument is required to that presented above.
\end{rem}

It remains to consider (\ref{eqn:inwards}), which reflects the role of the acceptance rate in the HMC method.  We turn to this next.

\subsubsection{Discussion of (\ref{eqn:inwards}).}

In \citep{roberts1996exponential} the authors note that (\ref{eqn:inwards}) applied to the MALA transition
$x_{\e}=x-\e^{2}\nabla U(x)/2+\e p_{0}$ can be
viewed as the restriction that for $x_{\e}\in I(x_{0})$
\[
U(x_0)-U(x_\e)-\hat{U}_{1}\geq\frac{\e^2}{8}\left(\|\nabla U(x_{\e})\|^{2}-\|\nabla U(x_{0})\|^{2}\right),
\]
where $\hat{U}_{1}:=\left\langle x_0 - x_\e,\nabla U(x_\e)+\nabla U(x_0)\right\rangle /2$
denotes the `trapezium' estimate for the line integral $U(x_0)-U(x_\e)=\int_{x_\e}^{x_0}\nabla U(z)dz$.
We can extend this intuition to HMC and arrive at the following natural
generalisation of the same condition.
\begin{prop}
The acceptance rate for HMC will satisfy the `inwards acceptance'
property (\ref{eqn:inwards}) if whenever $x_{L\e}\in I(x_{0})$ then in the limit as $\|x_0\|\to\infty$ it holds that
\begin{equation} \label{eqn:ic2}
U(x_0)-U(x_{L\e})-\hat{U}_{L}\geq\frac{1}{2L^{2}\e^{2}}\left(\|\psi_{L,\e}^{R}\|^{2}-\|\psi_{L,\e}\|^{2}\right),
\end{equation}
where $\hat{U}_L:=\langle x_{0}-x_{L\e},\nabla U(x_{0})+\nabla U(x_{L\e})+2\sum_{i=1}^{L-1}\nabla U(x_{i\e})\rangle/(2L)$
denotes the quadrature rule estimate for the line integral $\int_{x_0}^{x_{L\e}}\nabla U(z)dz$
based on $L$ trapezia and the forward and reverse drift components
are given by
\[
\psi_{L,\e}:=\frac{L\e^{2}}{2}\nabla U(x_{0})+\e^{2}\sum_{i=1}^{L-1}(L-i)\nabla U(x_{i\e}), ~~~~ 
\psi_{L,\e}^{R}:=\frac{L\e^{2}}{2}\nabla U(x_{L\e})+\e^{2}\sum_{i=1}^{L-1}i\nabla U(x_{i\e}).
\]
\end{prop}
\begin{proof}
We first note that we can write $p_{0}=\frac{1}{L\e}(x_{L\e}-x_{0}+\psi_{L,\e})$,
and that using reversibility of the leapfrog integrator, we can also
write $p_{L\e}=\frac{1}{L\e}(x_{L\e}-x_{0}-\psi_{L,\e}^{R})$.
The log acceptance ratio can therefore be written
\[
U(x_0)-U(x_{L\e})+\frac{1}{2L^{2}\e^{2}}\left(\|x_{L\e}-x_{0}+\psi_{L,\e}\|^{2}-\|x_{L\e}-x_{0}-\psi_{L,\e}^{R}\|^{2}\right).
\]
We require this quantity to be $\geq0$. This is equivalent to the
requirement
\[
U(x_{L\e})-U(x_0) \leq\frac{1}{2L^{2}\e^{2}}\left(\|x_{L\e}-x_{0}+\psi_{L,\e}\|^{2}-\|x_{L\e}-x_{0}-\psi_{L,\e}^{R}\|^{2}\right).
\]
We can re-write the right-hand side of the above expression as
\[
\frac{1}{2L^{2}\e^{2}}\left(2\langle x_{L\e}-x_{0},\psi_{L,\e}+\psi_{L,\e}^{R}\rangle+\|\psi_{L,\e}\|^{2}-\|\psi_{L,\e}^{R}\|^{2}\right),
\]
and then note that
\[
\psi_{L,\e}+\psi_{L,\e}^{R}=\frac{L\e^{2}}{2}\left(\nabla U(x_{0})+\nabla U(x_{L\e})+2\sum_{i=1}^{L-1}\nabla U(x_{i\e})\right).
\]
Substituting this into the inequality and simplifying gives the result. 
\end{proof}
The requirement (\ref{eqn:ic2}) can sometimes be established using convexity arguments. In the exponential family class of Corollary \ref{crly:expfam}, for example, setting $x^i:= x_{L\e} + i(x_0-x_{L\e})/L$, when $1\leq \beta<4/3$ one can show as $|x_0|\to\infty$ that $\hat{U}_L \xrightarrow{a.s.} (x_0-x_{L\e})(\nabla U(x_0) + \nabla U(x_{L\e}) + 2 \sum_{i=1}^{L-1} \nabla U(x^i))/(2L)$, the regular trapezium rule estimate for $\int_{x_{L\e}}^{x_0} \nabla U(z)dz = U(x_0) - U(x_{L\e})$. Since $\nabla U(x)$ is concave/convex for $x$ positive/negative then the trapezium rule gives an underestimate for the integral as $|x_0|\to\infty$, and hence the left-hand side of \eqref{eqn:ic2} will be positive, while it is also possible to show that the right hand side is negative in this case (using arguments given in the proof of Corollary \ref{crly:expfam}).  We omit the details of this.

There is some discussion in \citep{roberts1996exponential} of relaxations of (\ref{eqn:inwards}) to the requirement that $\alpha(x_{0},x_{\e})\geq \delta$ for some $\delta >0$ if $\|x_{\e}\|\leq\|x_{0}\|$,
which are also applicable to the HMC case and would relax the inequality
(\ref{eqn:ic2}) to some degree. In essence, the key role of the `inwards acceptance' property (\ref{eqn:inwards}) (among the
class of potentials which satisfy \textbf{A2}) is to limit the degree of
oscillation in the tails of the density $e^{-U(x)}$, which can potentially
mean that too many proposals $x_{L\e}$ for which the chosen
Lyapunov function $V(x_{L\e})/V(x_{0})<1$ are rejected to establish
a geometric bound of the form (\ref{eqn:geoerg}). Similar requirements to (\ref{eqn:inwards})
are needed for many Markov chain Monte Carlo methods which rely
on the Metropolis--Hastings construction (e.g. \citep{jarner2000geometric,roberts1996geometric,roberts1996exponential}).
The issues are discussed in some detail in the case of the Random
Walk Metropolis in \citep{jarner2000geometric}. It is possible that choosing the more natural
(but less pliable) Lyapunov function $V(x)=e^{sU(x)}$ for some $s>0$
would remove the need for (\ref{eqn:inwards}) here, owing to the ergodic nature
of the proposal kernel. We leave such explorations for future work.

The preceding discussion leads to the following assumption that we require for geometric ergodicity here.

\vspace{0.35cm}

\noindent \textbf{A3} {\itshape The chain satisfies the `inwards acceptance' property (\ref{eqn:inwards}) which can equivalently be formulated as (\ref{eqn:ic2}).} \label{ass:A3}

\vspace{0.35cm}

\noindent Assumptions \textbf{A1-A3} together are sufficient to establish a geometric bound.

\begin{proof}[Proof of Corollary \ref{crly:expfam}] \label{proof:expfam}

Part (ii) is a direct consequence of Theorem \ref{thm:negative} For part (i), we
consider three cases separately. 

First consider $\beta\in(1,2)$, meaning $1>\beta-1>0$. Since $\nabla U(x)=\alpha\beta\text{sgn}(x)|x|^{\beta-1}$
then \textbf{A2} holds. It remains to establish \textbf{A3}. We let
$x_{0}\to\infty$ but an analogous argument holds as $x_{0}\to-\infty$
by symmetry. Note that $x_{\varepsilon}=x_{0}-\varepsilon^{2}\alpha\beta\text{sgn}(x_{0})|x_{0}|^{\beta-1}/2+\varepsilon p_{0}$
will clearly satisfy $(1-\delta)x_{0}<x_{\varepsilon}<x_{0}$ with
probability reaching one in the limit for any $\delta>0$. Similarly
$(1-\delta)x_{0}<x_{2\varepsilon}=x_{\varepsilon}-\varepsilon^{2}\alpha\beta\text{sgn}(x_{0})|x_{0}|^{\beta-1}-\varepsilon^{2}\alpha\beta\text{sgn}(x_{\varepsilon})|x_{\varepsilon}|^{\beta-1}/2+\varepsilon p_{0}<x_{\varepsilon}$
in the same asymptotic regime. Iterating the argument reveals that
$(1-\delta)x_{0}<x_{L\varepsilon}<...<x_{\varepsilon}<x_{0}$ with
probability tending to one as $x_{0}\to\infty$. Hence a.s. the proposal
will be `inwards', as will each intermediate point in the trajectory.
To establish geometric convergence we must show that these inwards
proposals are accepted with probability tending to one as $x_{0}\to\infty$.
A Taylor series expansion of the difference in Hamiltonians for large enough $x_{0}$
gives
\[
H(x_{0},p_{0})-H(x_{L\varepsilon},p_{L\varepsilon})=\frac{L^{2}\varepsilon^{4}}{8}(\alpha\beta)^{3}(\beta-1)x_{0}^{3\beta-4}+o(x^{3\beta-4}).
\]

Since the leading order term is strictly positive then the result
is proved. A detailed derivation is provided in the supplementary
material \citep{supplement}.

In the case $\beta=1$ then as $x_{0}\to\infty$ the proposal in fact
a.s. becomes $x_{L\varepsilon}=x_{0}-L\varepsilon^{2}/2+L\varepsilon p_{0}$,
which resembles that of a random walk with inwards drift. Here the
acceptance rate a.s. becomes one as the leapfrog integrator becomes
exact provided the zero boundary is not crossed, and hence the scheme
is geometrically ergodic following Theorem 16.0.1 and the argument
of Section 16.1.3 in Chapter 16 of \cite{meyn2012markov}. Again
a similar argument holds as $x_{0}\to-\infty$.

In the case $\beta=2$ following Example 3.5 in \cite{beskos2013optimal},
setting $\theta:=\arccos(1-\alpha\varepsilon^{2})$ the proposal becomes
$x_{L\varepsilon}=\cos(\theta L)x_{0}+\sin(\theta L)p_{0}/\sqrt{2\alpha(1-\alpha\varepsilon^{2}/2)}$,which
will be inwards provided $|\cos(\theta L)|<1$, which will be true
for suitably small $\varepsilon$. Similarly provided $p_{0}=o(x_{0})$
then the difference in Hamiltonian values will be
\[
H(x_{0},p_{0})-H(x_{L\varepsilon},p_{L\varepsilon})=\left(1-\cos^{2}(\theta L)-2\alpha\left(1-\alpha\frac{\varepsilon^{2}}{2}\right)\sin^{2}(\theta L)\right)x_{0}^{2}+o(x_{0}^{2}).
\]
The $x_{0}^{2}$ coefficient will be positive provided $(1+2\alpha-\alpha^{2}\varepsilon^{2})\sin^{2}(\theta L)>0$
which will also be true for small enough $\varepsilon$, hence as
$x_{0}\to\pm\infty$ \textbf{A3} holds and since \textbf{A2} does
also then the result is proven.
\end{proof}

\subsection{Necessary conditions for geometric ergodicity}

Next we highlight the importance of the growth assumptions we have made on the potential, by showing two general scenarios in which HMC will \emph{not} produce geometrically ergodic Markov chains.

\subsubsection{Light tails}
We begin with the case where the gradient term may grow at a faster than linear rate, meaning that the resulting system of equations (\ref{eqn:hamilton}) is `stiff', in the sense that the derivatives can change very rapidly over small time scales, which can pose a challenge to explicit numerical integrators.  We show in Theorem \ref{thm:light} that in this scenario a Markov chain produced by the HMC method can exhibit undesirable behaviour.

\begin{thm} \label{thm:light}
If it holds that
\begin{equation} \label{eqn:explode}
\lim_{\|x\| \to \infty}\frac{\|\nabla U(x)\|}{\|x\|} = \infty,
\end{equation}
and that there is a fixed $C <\infty$ such that whenever $\|y\|\geq 2\|x\|\geq C$ then
\begin{equation} \label{eqn:oscillations}
\|\nabla U(y)\| \geq 3\|\nabla U(x)\|,
\end{equation} 
and it also holds that
\begin{equation} \label{eqn:oscillation2}
\lim_{\|x\|\to\infty, \|y\| \geq 2^L\|x\|} \left( U(x) - U(y) - \frac{1}{2}\|x\|^2 \right) = -\infty,
\end{equation}
then the Hamiltonian Monte Carlo method with fixed integration time $T = L\e$ does not produce a geometrically ergodic Markov chain for any choice $T>0$.
\end{thm}

\begin{proof} Lemmas \ref{lem:xelower} and \ref{lem:xlower} below establish that in this case $\|x_{L\e}\|\geq 2^L\|x_0\|$, and Lemma \ref{lem:arate_0} shows that this will result in $\alpha(x_0,x_{L\e})$ tending to zero as $\|x_0\|\to\infty$ provided $\|p_0\|\leq \|x_0\|^\delta$ for some $\delta<1$, allowing Proposition \ref{prop:supr} to be envoked.  To conclude we simply note that $\mathbb{P}(\|p_0\|\leq \|x_0\|^\delta) \to 1$ as $\|x_0\| \to \infty$, establishing the result.
\end{proof}

The conditions (\ref{eqn:oscillations}) and (\ref{eqn:oscillation2}) limit the amount that the potential can oscillate as it approaches $\infty$, and are introduced to prevent tail oscillations in gradient from making the behaviour of the method too unpredictable to analyse sensibly.  They are very lenient and should be satisfied for the vast majority of statistical models of interest for which (\ref{eqn:explode}) holds.  Below we establish several intermediate results, the first two of which relate to the values of $\|x_{L\e}\|$ when $\|x_0\|$ is large in this scenario.

\begin{lem} \label{lem:xelower}
If (\ref{eqn:explode}) holds then there exists an $\eta<\infty$ such that for all $\|x_0\|>\eta$ and any $\|p_0\|\leq \|x_0\|^\delta$ for some $\delta<1$, it holds that $\|x_\e\|>2\|x_0\|$.
\end{lem}

\begin{proof} Taking norms after a single leapfrog step gives
\begin{align*}
\|x_\e\| = \left\| x_0 - \frac{\e^2}{2}\nabla U(x_0) + \e p_0 \right\| &\geq \frac{\e^2}{2}\|\nabla U(x_0)\| - \|x_0\| - \e\|p_0\|.
\end{align*}
Dividing by $\|x_0\|$ gives
\begin{equation}
\frac{\|x_\e\|}{\|x_0\|} \geq \frac{\e^2}{2}\frac{\|\nabla U(x_0)\|}{\|x_0\|} - 1 - \e\frac{\|p_0\|}{\|x_0\|}.
\end{equation}
Using (\ref{eqn:explode}), we can choose an $x_0$ such that the first term on the right-hand side is larger than $6/\e^2$, and the last term can be made negligibly small as $\|p_0\| \leq \|x_0\|^\delta$ for some $\delta < 1$, which establishes the result.
\end{proof}

\begin{lem} \label{lem:xlower}
If (\ref{eqn:explode}) holds then there exists an $\eta<\infty$ such that for all $\|x_0\|>\eta$ and any $\|p_0\| \leq \|x_0\|^\delta$ for some $\delta < 1$, it holds that  $\|x_{L\e}\| \geq 2^L\|x_0\|$.
\end{lem}

\begin{proof}
We proceed iteratively.  First note that
\[
x_{2\e} = x_\e - \e^2\nabla U(x_\e) - \frac{\e^2}{2}\nabla U(x_0) + \e p_0.
\]
Using this, we have
\begin{align*}
\frac{\|x_{2\e}\|}{\|x_\e\|} \geq \frac{\e^2}{2}\frac{\|\nabla U(x_\e)\|}{\|x_\e\|} - \frac{\e^2}{2}\frac{\|\nabla U(x_0)\|}{\|x_\e \|} - 1 - \e\frac{\|p_0\|}{\|x_\e\|}.
\end{align*}
Showing the right-hand side is $\geq 2$ amounts to upper bounding the middle term, or equivalently lower bounding its reciprocal.  We have
\begin{equation} \label{eqn:inter21}
\frac{2 \|x_\e\|}{\e^2 \|\nabla U(x_0)\|} \geq \frac{ \e^2\|\nabla U(x_0)\|-2\|x_0\| - 2\e\|p_0\|}{\e^2\|\nabla U(x_0)\|} \geq 1 - \delta
\end{equation}
for some $\delta>0$ which can be made arbitrarily small by choosing $\|x_0\|$ large enough.

Next we have
\[
\frac{\|x_{3\e}\|}{\|x_{2\e}\|} \geq \frac{\e^2}{2}\frac{\|\nabla U(x_{2\e})\|}{\|x_{2\e}\|} - \frac{\e^2}{2}\frac{\|\nabla U(x_0)\|}{\|x_{2\e} \|} - \e^2\frac{\|\nabla U(x_\e)\|}{\|x_{2\e}\|} - 1 - \e\frac{\|p_0\|}{\|x_{2\e}\|}.
\]
Here the right-hand side will be $\geq 2$ provided the middle two terms can be bounded above. For the first we lower bound the reciprocal, using (\ref{eqn:inter21}) gives
\[
\frac{2}{\e^2}\frac{\|x_{2\e} \|}{\|\nabla U(x_0)\|} \geq \frac{2}{\e^2} \frac{ \|x_\e\|}{\|\nabla U(x_0)\|} \geq 1-\delta.
\]
For the second we have
\[
\frac{\|x_{2\e}\|}{\e^2\|\nabla U(x_\e)\|} \geq 1 - \frac{\|x_\e\|}{\e^2 \|\nabla U(x_\e)\|} - \frac{\|\nabla U(x_0)\|}{2\|\nabla U(x_\e)\|} - \frac{\|p_0\| }{\e \|\nabla U(x_\e)\|}.
\]
The second and last terms on the right hand side can be made arbitrarily small by choosing $\|x_0\|$ large enough.  Envoking (\ref{eqn:oscillations}) gives
\[
\frac{\|\nabla U(x_0)\|}{2\|\nabla U(x_\e)\|} \leq \frac{1}{6},
\]
which therefore shows that $\|x_{3\e}\|\geq 2\|x_{2\e}\|$.  An entirely analogous argument can be used to show that $\|x_{i\e}\|\geq 2\|x_{(i-1)\e}\|$ for any fixed $i$, establishing the result.
\end{proof}

The next result shows that as a result of the fact that $\|x_{L\e}\|\geq 2^L\|x_0\|$ when $\|x_0\|$ is large enough, then the acceptance rate will approach $0$ in the limit as $\|x_0\| \to \infty$.

\begin{lem} \label{lem:arate_0}
If (\ref{eqn:explode}), (\ref{eqn:oscillations}) and (\ref{eqn:oscillation2}) hold then for any $\delta<1$ it holds that
\[
\lim_{\|x_0\|\to\infty,\|p_0\|\leq\|x_0\|^\delta} \alpha(x_0,x_{L\e}) = 0.
\]
\end{lem}

\begin{proof}
Recall that
\begin{align*}
\alpha(x_0,x_{L\e}) = 1 \wedge \exp\left( U(x_0) - U(x_{L\e}) + \frac{1}{2}\|p_0\|^2 - \frac{1}{2}\|p_{L\e}\|^2 \right).
\end{align*}
Note that
\begin{align*}
\|p_{L\e}\| &\geq \frac{\e}{2}\|\nabla U(x_{L\e})\| - \e\sum_{i=1}^{L-1} \|\nabla U(x_{i\e})\| - \frac{\e}{2}\|\nabla U(x_0)\| - \|p_0\|, \\
&\geq \frac{\e}{2}\|\nabla U(x_{L\e})\| - \e\sum_{i=1}^{L} \left( \frac{1}{3} \right)^i \|\nabla U(x_{L\e})\| - \|p_0\|, \\
&= \frac{\e}{2}\left( 1- 2\sum_{i=1}^L \left(\frac{1}{3}\right)^i \right) \|\nabla U(x_{L\e})\| - \|p_0\|,
\end{align*}
where (\ref{eqn:oscillations}) is used for the second line.  The term inside the bracket can be bounded below by some fixed constant $\gamma_L >0$, for any fixed $L<\infty$.  Squaring the result gives
\begin{align*}
\|p_{L\e}\|^2 &\geq \left( \frac{\e\gamma_L}{2}\|\nabla U(x_{L\e})\| - \|p_0\|\right)^2 \\ 
&= \frac{\e^2\gamma_L^2}{4}\|\nabla U(x_{L\e})\|^2 + \|p_0\|^2 - \e\delta_L\|\nabla U(x_{L\e})\|\|p_0\|,
\end{align*}
which implies that
\[
\|p_0\|^2 - \|p_{L\e}\|^2 \leq \e\gamma_L\|\nabla U(x_{L\e})\| \left( \|p_0\| - \frac{\e\gamma_L}{4}\|\nabla U(x_{L\e})\| \right).
\]
Noting that $\|p_0\|\leq \|x_0\|^\delta$ and that for any $M<\infty$ we can choose an $\|x_0\|$ large enough that
\[
\frac{\e\gamma_L}{4}\|\nabla U(x_{L\e})\| \geq M\|x_{L\e}\|\geq 2^L\|x_0\|,
\]
then it follows that $\|p_0\|^2 - \|p_{L\e}\|^2 \leq -\|x_0\|^2$.  Using this, then simply envoking (\ref{eqn:oscillation2}) gives the result.
\end{proof}

\subsubsection{Heavy tails}

In the case where $\pi(x)$ has `heavier than exponential' tails in some direction the HMC method can also exhibit slow convergence, as $\liminf_{\|x\| \to \infty}\|\nabla U(x)\| = 0$.  Intuitively the problem here is that when $\|x\|$ is large then the gradient provides insufficient drift back into the `centre' of the space, meaning the chain can exhibit random walk behaviour and hence convergence can be very slow.  Theorem \ref{thm:heavy} makes this intuition rigorous.

\begin{thm} \label{thm:heavy}
If $\|\nabla U(x)\| <M$ for all $x \in \X$, then a \emph{necessary} condition for the Hamiltonian Monte Carlo method to produce a geometrically ergodic Markov chain is
\[
\int e^{s\|x\|}\pi(dx) < \infty
\]
for some $s>0$.
\end{thm}

\begin{proof}
From Proposition \ref{prop:tight}, it is sufficient to show that for any $\e>0$ there is a $\delta >0$ such that $Q(x,B_\delta(x)) > 1-\e$ for all $x \in \X$.  Using equation (\ref{eqn:mprop}) if $x_0$ is the current point in the chain then
\[
\|x_{L\e} - x_0\| = \left\| L\e p_0 - \frac{L\e^2}{2}\nabla U(x_0) - \e^2\sum_{i=1}^{L-1} (L-i)\nabla U(x_{i\e}) \right\|.
\]
Applying the triangle inequality and then the global bound on $\|\nabla U(x)\|$ gives
\[
\|x_{L\e} - x_0\| \leq \frac{L\e^2}{2}M + \frac{M\e^2 L(L-1)}{2} + L\e \|p_0\|.
\]
As $p_0$ follows a centred Gaussian distribution with fixed covariance then Chebyshev's inequality gives the result.
\end{proof}

In fact, in this case the lack of geometric ergodicity is a property of the flow itself, rather than being a consequence of numerical instabilities as in Theorem \ref{thm:light}, as shown by the following result.

\begin{prop} \label{prop:heavyexact}
Theorem \ref{thm:heavy} still holds even if an exact integrator is available for Hamilton's equations.
\end{prop}

\begin{proof}
Using Hamilton's equations, we have
\begin{equation}
x_T - x_0 = \int_0^T p_s ds = \int_0^T \left[ p_0 - \int_0^s \nabla U(x_u) du \right] ds.
\end{equation}
Taking the norm and using the upper bound gives
\begin{align*}
\|x_T - x_0\| \leq T\|p_0\| - \int_0^T \int_0^s \|\nabla U(x_u)\| duds \leq T\|p_0\| + CT^2/2,
\end{align*}
and again Chebyshev's inequality gives the result.
\end{proof}

The class of models for which $\|\nabla U(x)\|$ is bounded and $\int e^s\|x\|\pi(dx)<\infty$ for some $s>0$ is comparatively narrow, essentially comprising $U(x) = C\|x\| + b(x)$, where $C<\infty$ and $b:\X\to\R$ is some appropriately regular function which is bounded both above and below.

\section{Results for an position-dependent integration time}
\label{sec:dynamic}

An important free parameter in HMC is the integration time $T$, which we have previously assumed to be independent of the current position.  The representation (\ref{eqn:mprop}) does however suggest that allowing this to change can have some benefits.  If the candidate map is viewed as
\[
x_{L\e} = x_0 + \text{DRIFT}(x_0,p_0,T) + L\e p_0,
\]
then if the `DRIFT' function becomes negligible for large $\|x_0\|$ and fixed $T$, then it can be increased in magnitude by making $T$ larger.  We make this simple intuition rigorous for an idealised algorithm on the particular one-dimensional Exponential Family class of models with densities of the form
\begin{equation} \label{eqn:efsmooth}
\pi(x) \propto \exp \left( - \beta^{-1} (1+x^2) ^{\beta / 2} \right),
\end{equation}
for some fixed $\beta>0$.  Here any contour $C_{x_0,p_0} := \{ (x,p) : H(x,p) = H(x_0,p_0)\}$ consists of a single closed path, and the flow is periodic from any fixed starting point.  We additionally assume that the period length $\zeta_{x_0,p_0} >0$ is known, and that we have an exact integrator for Hamilton's equations.  This means that we need not concern ourselves with the acceptance probability (we discuss this issue in Section \ref{sec:discussion}).

At iteration $i$ (with $x_0 = x_{i-1}$), the \emph{dynamic} HMC implementation we consider consists of re-sampling $p_0 \sim N(0,1)$, and then setting $x_i = \text{Pr}_x \circ \varphi_{\tau}(x_0,p_0)$, where $\tau \sim U[0,\zeta_{x_0,p_0}]$.  In words, we flow along the Hamiltonian for $\tau$ units of time, where $\tau$ is a uniform random variable with maximum value $\zeta_{x_0,p_0}$ (note that $\varphi_{\zeta_{x_0,p_0}}(x_0,p_0) = (x_0,p_0)$).

Firstly, note that $\pi$-irreducibility is more straightforward to see here.  To reach any set $A \in \B$ with $\pi(A) > 0$, we first consider the single contour $C_{x_0,p_0}$, and specifically the component of this contour that is connected to $(x_0,p_0)$.  Let $C_{x_0}$ be the projection of this component onto $\X$.  Then any nonempty set $A' \subset C_{x_0}$ has positive probability of occuring, as the next point is chosen from a density with support all of $C_{x_0}$.  As the contours are composed of single components, and cover the entire space, then for any $A$, the probability of choosing a contour for which this argument can be applied is greater than zero.  We provide a figure in the supplementary material to offer more intuition \citep{supplement}.

We introduce some additional notation in this section.  We define the \emph{microcanonical} expectation of a real-valued function $f(x_t,p_t)$, where $(x_t,p_t) = \varphi_t(x_0,p_0)$, i.e. the solution to (\ref{eqn:hamilton}) for $t$ units of time initialised at $(x_0,p_0)$, as
\begin{equation}
\left< f(x_0,p_0) \right> := \zeta_{x_0,p_0}^{-1} \int_0^{\zeta_{x_0,p_0}} f(x_s,p_s)ds.
\end{equation}
This is simply the time expectation of $f$ from uniformly sampling across $C_{x_0,p_0}$.

We first introduce a result from the Physics literature (e.g. \citep{goldstein1965classical}) which relates the kinetic and potential energies.

\begin{thm} (Virial Theorem). \label{thm:virial}
Under Hamiltonian flow $(x_{s},p_{s}) = \varphi_s(x_0,p_0)$ we have
\begin{equation} \label{eqn:vt1}
\left< x_0\nabla U(x_0) \right> = \left< p_0^2 \right>.
\end{equation}
\end{thm}

\begin{proof} Define the \emph{virial} function $G_t = x_tp_t$.  From the fundamental theorem of Calculus we have
\[
\left< \dot{G}_0 \right> = \frac{G_{\zeta_{x_0,p_0}} - G_0}{\zeta_{x_0,p_0}} = 0,
\]
where $\dot{G}_t := dG_t/dt$.  In this case
\[
\dot{G}_t = x_t \dot{p}_t + p_t\dot{x}_t = -x_t \nabla U(x_t) + p_t^2,
\]
meaning
\[
\left< x_0 \nabla U(x_0) \right> = \left< p_0^2 \right>,
\]
as required.
\end{proof}

We can now state and prove the main result of this section.

\begin{thm} \label{thm:dynamic}
For the one-dimensional Exponential Family class of distributions with density given by (\ref{eqn:efsmooth}), the dynamic Hamiltonian Monte Carlo method produces a geometrically ergodic Markov chain for any value of $\beta > 0$.
\end{thm}

\begin{proof}
Note that by conservation of the Hamiltonian, we have
\begin{align} \label{eqn:conserveh}
\int \left< U(x_0) + p_0^2/2 \right> \mu^G(dp_0)  &= \int H(x_0,p_0)\mu^G(dp_0) = U(x_0) + 1/2,
\end{align}
Choose the Lyapunov function $V(x) = U(x) + xU'(x) + 1$.  Using Theorem \ref{thm:virial}, we can re-write the above expression
\[
PV(x_0)  = U(x_0) + 3/2.
\]
Note also that for any $\eta > 0$ there is an $M_\eta<\infty$ such that whenever $|x_0|>M_\eta$
\[
(1+\eta)U(x_0) \geq U(x_0) + 3/2.
\]
The proof will be complete if we can find a $\lambda < 1$ such that $(1+\eta)U(x_0) \leq \lambda V(x_0)$ for suitably large $|x_0|$.  Now $x_0U'(x_0) \to \beta U(x_0)$ here as $|x_0| \to \infty$, meaning that there is an $M<\infty$ such that whenever $|x_0|>M$
\[
x_0 U'(x_0)/2 \geq \beta U(x_0)/2 - 1.
\]
Taking $|x_0|\geq \max(M_\eta,M)$ we can therefore re-write the inequality of interest $(1+\eta)U(x_0) \leq \lambda V(x_0)$ as
\[
(1+\eta)U(x_0) \leq \lambda (1+\beta/2) U(x),
\]
which will be true if
\[
\lambda \geq \frac{1+\eta}{1+\beta/2}.
\]
Choosing $\eta<\beta/2$ ensures $\lambda <1$ and also gives the desired inequality $PV(x_0) \leq \lambda V(x_0)$ whenever $|x_0| > \max(M_\eta,M)$, showing that the resulting Markov chain will be geometrically ergodic.
\end{proof}

\section{Discussion}
\label{sec:discussion}

We have established conditions under which geometric ergodicity will and will not hold for Markov chains produced by the Hamiltonian Monte Carlo method.  Here we discuss how our results can be extended in various ways, as well as how they translate to standard implementations in widely used software \citep{carpenter2016stan}.

\subsection{Dynamic implementations}

Allowing the integration time in HMC to depend on the current point in the chain without an exact integrator will typically mean that some adjustments to $\alpha(x_0,x_T)$ must be made to ensure that $\pi(\cdot)$ is still preserved.  The reason is that the approximate flow map $\varphi_T$ may no longer be reversible, as if $T_1 := T(x_0,p_0)$ and $T_2 := T(x_T,p_T)$ then $\varphi^{-1}_{T_2} \circ \varphi_{T_1}$ will typically not be the identity map if $T_1 \neq T_2$.  The two possible ways of changing the integration time $T = L\e$ are to adjust either $L$ or $\e$.  Increasing $L$ requires more computations per transition, while this is not necessarily true for $\e$.  In the No-U-Turn sampler a binary tree approach is introduced to ensure preservation of detailed balance when $L$ is altered in different parts of the space \citep{hoffman-gelman:2013}.  We are not aware of any implementations involving adjustment of $\e$, however it is likely that similar modifications to $\alpha(x_0,x_T)$ are possible here also.  Adjusting $\e$ may be a sensible option in some cases, as the leapfrog method is known to `almost' preserve the modified Hamiltonian
\[
\tilde{H}(x,p) = H(x,p) + \left( \frac{1}{12} p^t \nabla^t\nabla U(x) p - \frac{1}{24}\nabla U(x)^t \nabla U(x) \right) \e^2 + O(\e^4),
\]
as shown for example in \citep{leimkuhler2004simulating}.  When $\pi(x)$ is not log-concave in the tails and hence the elements of $\nabla U$ and $\nabla^t \nabla U$ become negligible as $\|x\| \to \infty$, this implies that $\e$ can be increased for larger $\|x\|$ without compromising on numerical accuracy.

\subsection{Extension to other integrators}

The fixed integration time results in Section \ref{sec:fixed} refer specifically to the leapfrog integrator implementation of HMC (aside from Proposition \ref{prop:heavyexact}).  It should be possible to use the same approach when analysing other explicit symplectic integrators, however for schemes which rely on implicit methods (e.g. \citep{girolami2011riemann}) then composing multiple steps of the integrator as in Proposition \ref{prop:mprop} cannot be done so cleanly.  Implicit methods are needed when the Hamiltonian is non-separable, and can often resolve stiffness issues such as those characterised in Theorem \ref{thm:light}.

To construct ergodicity results for the most general version of Hamiltonian Monte Carlo (i.e. without restricting attention to $\mathbb{R}^n$) we note that there are many ways to construct drift conditions in line with the purely geometric framework introduced in \citep{betancourt2014geometric}.  We also point out that for the one-dimensional Exponential Family, choosing the Riemannian metric $G(x) = \|\nabla^2U(x)\|$ and employing the approach of \citep{girolami2011riemann} is mathematically equivalent to applying the transformation $x' = \text{sgn}(x)\|x\|^{\beta/2}$ with corresponding density
\[
\pi(x') \propto \|x'\|^{2/(\beta - 1)}\exp \left( -x'^2/\beta^2 \right).
\]
This new density will have Gaussian tails for any $\beta >0$, suggesting a well-behaved sampler can be constructed.  Further discussion on the relationship between geometric Markov chain Monte Carlo methods and parameter transformations is given in \citep{livingstone2014information}.

\subsection{Honest bounds}

Geometric ergodicity is often called a \emph{qualitative} bound, as an explicit upper bound on the geometric rate $\rho$ is not established when using the techniques of \citep{roberts1996geometric}.  With some modifications, however, quantitative bounds can be constructed (e.g. \citep{jones2001honest}).  We have refrained from doing this here, as these bounds are also often too conservative to be of use in practice \citep{jones2001honest}.

Monte Carlo estimates for non-asymptotic quantitative bounds using the Ricci curvature approach of \citep{ollivier2009ricci} are applied to Hamiltonian Monte Carlo in \citep{seiler2014positive}.  We note that the applicability of these bounds relies on the assumption of \emph{positive curvature} in some Wasserstein distance for the underlying Markov chain.  When this distance is chosen to be Total Variation, then this is a strictly stronger condition than geometric ergodicity [see Corollary 22 in \citep{ollivier2009ricci}], so we feel that our results are a useful pre-cursor to understanding when these estimated bounds are informative in practice.

In the case of MALA, when $\|\nabla U(x)\|$ grows at a faster than linear rate for large $\|x\|$ then it is shown in \citep{bou2012nonasymptotic} that useful inferences for functionals concentrated in the centre of the space can be made by setting a small enough value for $\e$.  It is likely that the same analysis can be done with HMC, and that the result would be similar, but we leave such explorations for future work.

\subsection{Practitioner guidelines}

The main conclusion of our work for practitioners implementing the method in a bespoke manner is to consider the form of $\|\nabla U(x)\|$.  If this term either grows very fast or becomes negligibly small when $\|x\|$ is large then it is likely that the Markov chains produced will struggle to explore the tails of $\pi(\cdot)$ effectively.  When the gradient grows at a faster than linear rate then a suitably small value for $\e$ must be chosen to counteract this, while when it shrinks then the integration time $T$ must be made sufficiently large.  Of course in either scenario if there is a re-parametrisation of the model that may not suffer these difficulties then this should be applied.  Users implementing the method in the Stan software \citep{carpenter2016stan} should note that both of these instances are captured by standard output diagnostics.  Numerical trajectories that become unstable due to large gradients are classed as `divergences', while a failure to move far enough because of negligible gradients is recorded through the `maximum tree depth reached' warning.  If this happens and $\pi(\cdot)$ is known to be proper then the user should set as large a maximum tree depth as is computationally feasible when tail exploration is of keen interest.

\section*{Acknowledgements}

We thank Alexandros Beskos, Gareth Roberts, Krzysztof \L{}atuszy\'{n}ski, Gabriel Stoltz and Mark Rowland for useful discussions.  SL thanks Nawaf Bou--Rabee for pointing him to \citep{cances2007theoretical}.

SL was supported by a PhD scholarship from Xerox Research Centre Europe and EPSRC grant EP/K014463/1 for this project.  SB was supported by EPSRC fellowship EP/K005723/1, MB is funded by EPSRC grant EP/J016934/1, and SB and MB also received a 2014 EPRSC NCSML Award for PDRA Collaboration for this project.  MG is funded by an EPSRC Established Career Research Fellowship, EP/J016934/1, a Royal Society Wolfson Research Merit Award, and EPSRC grants EP/P020720/1, EP/J016934/3, EP/K034154/1.

\bibliographystyle{imsart-nameyear}
\bibliography{mybibliography}

\begin{thebibliography}{48}

\bibitem[\protect\citeauthoryear{Alder and Wainwright}{1959}]{alder1959studies}
\begin{barticle}[author]
\bauthor{\bsnm{Alder},~\bfnm{Berni~J}\binits{B.~J.}} \AND
  \bauthor{\bsnm{Wainwright},~\bfnm{TE}\binits{T.}}
(\byear{1959}).
\btitle{Studies in molecular dynamics. I. General method}.
\bjournal{The Journal of Chemical Physics}
\bvolume{31}
\bpages{459--466}.
\end{barticle}
\endbibitem

\bibitem[\protect\citeauthoryear{Andrieu
  et~al.}{2003}]{andrieu2003introduction}
\begin{barticle}[author]
\bauthor{\bsnm{Andrieu},~\bfnm{Christophe}\binits{C.}},
  \bauthor{\bsnm{De~Freitas},~\bfnm{Nando}\binits{N.}},
  \bauthor{\bsnm{Doucet},~\bfnm{Arnaud}\binits{A.}} \AND
  \bauthor{\bsnm{Jordan},~\bfnm{Michael~I}\binits{M.~I.}}
(\byear{2003}).
\btitle{An introduction to MCMC for machine learning}.
\bjournal{Machine learning}
\bvolume{50}
\bpages{5--43}.
\end{barticle}
\endbibitem

\bibitem[\protect\citeauthoryear{Beskos et~al.}{2011}]{beskos2011hybrid}
\begin{barticle}[author]
\bauthor{\bsnm{Beskos},~\bfnm{Alexandros}\binits{A.}},
  \bauthor{\bsnm{Pinski},~\bfnm{Frank~J}\binits{F.~J.}},
  \bauthor{\bsnm{Sanz-Serna},~\bfnm{Jes{\'u}s~Mar{\i}a}\binits{J.~M.}} \AND
  \bauthor{\bsnm{Stuart},~\bfnm{Andrew~M}\binits{A.~M.}}
(\byear{2011}).
\btitle{Hybrid monte carlo on hilbert spaces}.
\bjournal{Stochastic Processes and their Applications}
\bvolume{121}
\bpages{2201--2230}.
\end{barticle}
\endbibitem

\bibitem[\protect\citeauthoryear{Beskos et~al.}{2013}]{beskos2013optimal}
\begin{barticle}[author]
\bauthor{\bsnm{Beskos},~\bfnm{Alexandros}\binits{A.}},
  \bauthor{\bsnm{Pillai},~\bfnm{Natesh}\binits{N.}},
  \bauthor{\bsnm{Roberts},~\bfnm{Gareth}\binits{G.}},
  \bauthor{\bsnm{Sanz-Serna},~\bfnm{Jes{\'u}s~Mar{\i}a}\binits{J.~M.}} \AND
  \bauthor{\bsnm{Stuart},~\bfnm{Andrew}\binits{A.}}
(\byear{2013}).
\btitle{Optimal tuning of the hybrid Monte Carlo algorithm}.
\bjournal{Bernoulli}
\bvolume{19}
\bpages{1501--1534}.
\end{barticle}
\endbibitem

\bibitem[\protect\citeauthoryear{Betancourt}{2013}]{betancourt2013general}
\begin{bincollection}[author]
\bauthor{\bsnm{Betancourt},~\bfnm{Michael}\binits{M.}}
(\byear{2013}).
\btitle{A general metric for Riemannian manifold Hamiltonian Monte Carlo}.
In \bbooktitle{Geometric science of information}
\bpages{327--334}.
\bpublisher{Springer}.
\end{bincollection}
\endbibitem

\bibitem[\protect\citeauthoryear{Betancourt}{2016}]{betancourt2016identifying}
\begin{barticle}[author]
\bauthor{\bsnm{Betancourt},~\bfnm{Michael}\binits{M.}}
(\byear{2016}).
\btitle{Identifying the Optimal Integration Time in Hamiltonian Monte Carlo}.
\bjournal{arXiv preprint arXiv:1601.00225}.
\end{barticle}
\endbibitem

\bibitem[\protect\citeauthoryear{Betancourt
  et~al.}{2016}]{betancourt2014geometric}
\begin{barticle}[author]
\bauthor{\bsnm{Betancourt},~\bfnm{MJ}\binits{M.}},
  \bauthor{\bsnm{Byrne},~\bfnm{Simon}\binits{S.}},
  \bauthor{\bsnm{Livingstone},~\bfnm{Samuel}\binits{S.}} \AND
  \bauthor{\bsnm{Girolami},~\bfnm{Mark}\binits{M.}}
(\byear{2016}).
\btitle{The Geometric Foundations of Hamiltonian Monte Carlo}.
\bjournal{Bernoulli}
\bvolume{forthcoming}.
\end{barticle}
\endbibitem

\bibitem[\protect\citeauthoryear{Bou-Rabee and
  Hairer}{2012}]{bou2012nonasymptotic}
\begin{barticle}[author]
\bauthor{\bsnm{Bou-Rabee},~\bfnm{Nawaf}\binits{N.}} \AND
  \bauthor{\bsnm{Hairer},~\bfnm{Martin}\binits{M.}}
(\byear{2012}).
\btitle{Nonasymptotic mixing of the MALA algorithm}.
\bjournal{IMA Journal of Numerical Analysis}
\bpages{drs003}.
\end{barticle}
\endbibitem

\bibitem[\protect\citeauthoryear{Bou-Rabee and
  Sanz-Serna}{2015}]{bou2015randomized}
\begin{barticle}[author]
\bauthor{\bsnm{Bou-Rabee},~\bfnm{Nawaf}\binits{N.}} \AND
  \bauthor{\bsnm{Sanz-Serna},~\bfnm{Jes{\'u}s~Mar{\i}a}\binits{J.~M.}}
(\byear{2015}).
\btitle{Randomized Hamiltonian Monte Carlo}.
\bjournal{arXiv preprint arXiv:1511.09382}.
\end{barticle}
\endbibitem

\bibitem[\protect\citeauthoryear{Bou-Rabee and
  Sanz-Serna}{2018}]{bou2018geometric}
\begin{barticle}[author]
\bauthor{\bsnm{Bou-Rabee},~\bfnm{Nawaf}\binits{N.}} \AND
  \bauthor{\bsnm{Sanz-Serna},~\bfnm{Jes{\'u}s~Mar{\'\i}a}\binits{J.~M.}}
(\byear{2018}).
\btitle{Geometric integrators and the Hamiltonian Monte Carlo method}.
\bjournal{Acta Numerica}
\bvolume{27}
\bpages{113--206}.
\end{barticle}
\endbibitem

\bibitem[\protect\citeauthoryear{Byrne and Girolami}{2013}]{byrne2013geodesic}
\begin{barticle}[author]
\bauthor{\bsnm{Byrne},~\bfnm{Simon}\binits{S.}} \AND
  \bauthor{\bsnm{Girolami},~\bfnm{Mark}\binits{M.}}
(\byear{2013}).
\btitle{Geodesic Monte Carlo on embedded manifolds}.
\bjournal{Scandinavian Journal of Statistics}
\bvolume{40}
\bpages{825--845}.
\end{barticle}
\endbibitem

\bibitem[\protect\citeauthoryear{Campos and Sanz-Serna}{2015}]{campos2015extra}
\begin{barticle}[author]
\bauthor{\bsnm{Campos},~\bfnm{C{\'e}dric~M}\binits{C.~M.}} \AND
  \bauthor{\bsnm{Sanz-Serna},~\bfnm{JM}\binits{J.}}
(\byear{2015}).
\btitle{Extra chance generalized hybrid Monte Carlo}.
\bjournal{Journal of Computational Physics}
\bvolume{281}
\bpages{365--374}.
\end{barticle}
\endbibitem

\bibitem[\protect\citeauthoryear{Canc{\`e}s, Legoll and
  Stoltz}{2007}]{cances2007theoretical}
\begin{barticle}[author]
\bauthor{\bsnm{Canc{\`e}s},~\bfnm{Eric}\binits{E.}},
  \bauthor{\bsnm{Legoll},~\bfnm{Fr{\'e}d{\'e}ric}\binits{F.}} \AND
  \bauthor{\bsnm{Stoltz},~\bfnm{Gabriel}\binits{G.}}
(\byear{2007}).
\btitle{Theoretical and numerical comparison of some sampling methods for
  molecular dynamics}.
\bjournal{ESAIM: Mathematical Modelling and Numerical Analysis}
\bvolume{41}
\bpages{351--389}.
\end{barticle}
\endbibitem

\bibitem[\protect\citeauthoryear{Carpenter et~al.}{2016}]{carpenter2016stan}
\begin{barticle}[author]
\bauthor{\bsnm{Carpenter},~\bfnm{Bob}\binits{B.}},
  \bauthor{\bsnm{Gelman},~\bfnm{Andrew}\binits{A.}},
  \bauthor{\bsnm{Hoffman},~\bfnm{Matt}\binits{M.}},
  \bauthor{\bsnm{Lee},~\bfnm{Daniel}\binits{D.}},
  \bauthor{\bsnm{Goodrich},~\bfnm{Ben}\binits{B.}},
  \bauthor{\bsnm{Betancourt},~\bfnm{Michael}\binits{M.}},
  \bauthor{\bsnm{Brubaker},~\bfnm{Michael~A}\binits{M.~A.}},
  \bauthor{\bsnm{Guo},~\bfnm{Jiqiang}\binits{J.}},
  \bauthor{\bsnm{Li},~\bfnm{Peter}\binits{P.}} \AND
  \bauthor{\bsnm{Riddell},~\bfnm{Allen}\binits{A.}}
(\byear{2016}).
\btitle{Stan: A probabilistic programming language}.
\bjournal{Journal of Statistical Software}.
\end{barticle}
\endbibitem

\bibitem[\protect\citeauthoryear{Diaconis}{2013}]{diaconis2013some}
\begin{barticle}[author]
\bauthor{\bsnm{Diaconis},~\bfnm{Persi}\binits{P.}}
(\byear{2013}).
\btitle{Some things we’ve learned (about Markov chain Monte Carlo)}.
\bjournal{Bernoulli}
\bvolume{19}
\bpages{1294--1305}.
\end{barticle}
\endbibitem

\bibitem[\protect\citeauthoryear{Diaconis and
  Freedman}{1999}]{diaconis1999iterated}
\begin{barticle}[author]
\bauthor{\bsnm{Diaconis},~\bfnm{Persi}\binits{P.}} \AND
  \bauthor{\bsnm{Freedman},~\bfnm{David}\binits{D.}}
(\byear{1999}).
\btitle{Iterated random functions}.
\bjournal{SIAM review}
\bvolume{41}
\bpages{45--76}.
\end{barticle}
\endbibitem

\bibitem[\protect\citeauthoryear{Diaconis, Seiler and
  Holmes}{2014}]{diaconis2014connections}
\begin{barticle}[author]
\bauthor{\bsnm{Diaconis},~\bfnm{Persi}\binits{P.}},
  \bauthor{\bsnm{Seiler},~\bfnm{Christof}\binits{C.}} \AND
  \bauthor{\bsnm{Holmes},~\bfnm{Susan}\binits{S.}}
(\byear{2014}).
\btitle{Connections and Extensions: A Discussion of the Paper by Girolami and
  Byrne}.
\bjournal{Scandinavian Journal of Statistics}
\bvolume{41}
\bpages{3--7}.
\end{barticle}
\endbibitem

\bibitem[\protect\citeauthoryear{Duane et~al.}{1987}]{duane1987hybrid}
\begin{barticle}[author]
\bauthor{\bsnm{Duane},~\bfnm{Simon}\binits{S.}},
  \bauthor{\bsnm{Kennedy},~\bfnm{Anthony~D}\binits{A.~D.}},
  \bauthor{\bsnm{Pendleton},~\bfnm{Brian~J}\binits{B.~J.}} \AND
  \bauthor{\bsnm{Roweth},~\bfnm{Duncan}\binits{D.}}
(\byear{1987}).
\btitle{Hybrid monte carlo}.
\bjournal{Physics letters B}
\bvolume{195}
\bpages{216--222}.
\end{barticle}
\endbibitem

\bibitem[\protect\citeauthoryear{Durmus and
  Moulines}{2015}]{durmus2015quantitative}
\begin{barticle}[author]
\bauthor{\bsnm{Durmus},~\bfnm{Alain}\binits{A.}} \AND
  \bauthor{\bsnm{Moulines},~\bfnm{{\'E}ric}\binits{{\'E}.}}
(\byear{2015}).
\btitle{Quantitative bounds of convergence for geometrically ergodic Markov
  chain in the Wasserstein distance with application to the Metropolis Adjusted
  Langevin Algorithm}.
\bjournal{Statistics and Computing}
\bvolume{25}
\bpages{5--19}.
\end{barticle}
\endbibitem

\bibitem[\protect\citeauthoryear{Durmus, Moulines and
  Saksman}{2017}]{durmus2017convergence}
\begin{barticle}[author]
\bauthor{\bsnm{Durmus},~\bfnm{Alain}\binits{A.}},
  \bauthor{\bsnm{Moulines},~\bfnm{Eric}\binits{E.}} \AND
  \bauthor{\bsnm{Saksman},~\bfnm{Eero}\binits{E.}}
(\byear{2017}).
\btitle{On the convergence of Hamiltonian Monte Carlo}.
\bjournal{arXiv preprint arXiv:1705.00166}.
\end{barticle}
\endbibitem

\bibitem[\protect\citeauthoryear{Eberle}{2014}]{eberle2014error}
\begin{barticle}[author]
\bauthor{\bsnm{Eberle},~\bfnm{Andreas}\binits{A.}}
(\byear{2014}).
\btitle{Error bounds for Metropolis--Hastings algorithms applied to
  perturbations of Gaussian measures in high dimensions}.
\bjournal{The Annals of Applied Probability}
\bvolume{24}
\bpages{337--377}.
\end{barticle}
\endbibitem

\bibitem[\protect\citeauthoryear{Gelman et~al.}{2014}]{gelman2014bayesian}
\begin{bbook}[author]
\bauthor{\bsnm{Gelman},~\bfnm{Andrew}\binits{A.}},
  \bauthor{\bsnm{Carlin},~\bfnm{John~B}\binits{J.~B.}},
  \bauthor{\bsnm{Stern},~\bfnm{Hal~S}\binits{H.~S.}} \AND
  \bauthor{\bsnm{Rubin},~\bfnm{Donald~B}\binits{D.~B.}}
(\byear{2014}).
\btitle{Bayesian data analysis}
\bvolume{2}.
\bpublisher{Taylor \& Francis}.
\end{bbook}
\endbibitem

\bibitem[\protect\citeauthoryear{Girolami and
  Calderhead}{2011}]{girolami2011riemann}
\begin{barticle}[author]
\bauthor{\bsnm{Girolami},~\bfnm{Mark}\binits{M.}} \AND
  \bauthor{\bsnm{Calderhead},~\bfnm{Ben}\binits{B.}}
(\byear{2011}).
\btitle{Riemann manifold langevin and hamiltonian monte carlo methods}.
\bjournal{Journal of the Royal Statistical Society: Series B (Statistical
  Methodology)}
\bvolume{73}
\bpages{123--214}.
\end{barticle}
\endbibitem

\bibitem[\protect\citeauthoryear{Goldstein}{1965}]{goldstein1965classical}
\begin{bbook}[author]
\bauthor{\bsnm{Goldstein},~\bfnm{Herbert}\binits{H.}}
(\byear{1965}).
\btitle{Classical mechanics}.
\bpublisher{Pearson Education India}.
\end{bbook}
\endbibitem

\bibitem[\protect\citeauthoryear{Hastings}{1970}]{hastings1970monte}
\begin{barticle}[author]
\bauthor{\bsnm{Hastings},~\bfnm{W~Keith}\binits{W.~K.}}
(\byear{1970}).
\btitle{Monte Carlo sampling methods using Markov chains and their
  applications}.
\bjournal{Biometrika}
\bvolume{57}
\bpages{97--109}.
\end{barticle}
\endbibitem

\bibitem[\protect\citeauthoryear{Hoffman and
  Gelman}{2014}]{hoffman-gelman:2013}
\begin{barticle}[author]
\bauthor{\bsnm{Hoffman},~\bfnm{Matthew~D.}\binits{M.~D.}} \AND
  \bauthor{\bsnm{Gelman},~\bfnm{Andrew}\binits{A.}}
(\byear{2014}).
\btitle{The No-{U}-Turn Sampler: Adaptively Setting Path Lengths in
  {H}amiltonian {M}onte {C}arlo}.
\bjournal{Journal of Machine Learning Research}
\bvolume{15}
\bpages{1593--1623}.
\end{barticle}
\endbibitem

\bibitem[\protect\citeauthoryear{Horowitz}{1991}]{horowitz1991generalized}
\begin{barticle}[author]
\bauthor{\bsnm{Horowitz},~\bfnm{Alan~M}\binits{A.~M.}}
(\byear{1991}).
\btitle{A generalized guided Monte Carlo algorithm}.
\bjournal{Physics Letters B}
\bvolume{268}
\bpages{247--252}.
\end{barticle}
\endbibitem

\bibitem[\protect\citeauthoryear{Jarner and Hansen}{2000}]{jarner2000geometric}
\begin{barticle}[author]
\bauthor{\bsnm{Jarner},~\bfnm{S{\o}ren~Fiig}\binits{S.~F.}} \AND
  \bauthor{\bsnm{Hansen},~\bfnm{Ernst}\binits{E.}}
(\byear{2000}).
\btitle{Geometric ergodicity of Metropolis algorithms}.
\bjournal{Stochastic processes and their applications}
\bvolume{85}
\bpages{341--361}.
\end{barticle}
\endbibitem

\bibitem[\protect\citeauthoryear{Jarner and
  Tweedie}{2003}]{jarner2003necessary}
\begin{barticle}[author]
\bauthor{\bsnm{Jarner},~\bfnm{S{\o}ren~F}\binits{S.~F.}} \AND
  \bauthor{\bsnm{Tweedie},~\bfnm{Richard~L}\binits{R.~L.}}
(\byear{2003}).
\btitle{Necessary conditions for geometric and polynomial ergodicity of
  random-walk-type}.
\bjournal{Bernoulli}
\bvolume{9}
\bpages{559--578}.
\end{barticle}
\endbibitem

\bibitem[\protect\citeauthoryear{Jones and Hobert}{2001}]{jones2001honest}
\begin{barticle}[author]
\bauthor{\bsnm{Jones},~\bfnm{Galin~L}\binits{G.~L.}} \AND
  \bauthor{\bsnm{Hobert},~\bfnm{James~P}\binits{J.~P.}}
(\byear{2001}).
\btitle{Honest exploration of intractable probability distributions via Markov
  chain Monte Carlo}.
\bjournal{Statistical Science}
\bpages{312--334}.
\end{barticle}
\endbibitem

\bibitem[\protect\citeauthoryear{Lee}{2012}]{lee2012symplectic}
\begin{bbook}[author]
\bauthor{\bsnm{Lee},~\bfnm{John~M}\binits{J.~M.}}
(\byear{2012}).
\btitle{Introduction to Smooth Manifolds}.
\bpublisher{Springer}.
\end{bbook}
\endbibitem

\bibitem[\protect\citeauthoryear{Leimkuhler and
  Reich}{2004}]{leimkuhler2004simulating}
\begin{bbook}[author]
\bauthor{\bsnm{Leimkuhler},~\bfnm{Benedict}\binits{B.}} \AND
  \bauthor{\bsnm{Reich},~\bfnm{Sebastian}\binits{S.}}
(\byear{2004}).
\btitle{Simulating hamiltonian dynamics}
\bvolume{14}.
\bpublisher{Cambridge University Press}.
\end{bbook}
\endbibitem

\bibitem[\protect\citeauthoryear{Livingstone and
  Girolami}{2014}]{livingstone2014information}
\begin{barticle}[author]
\bauthor{\bsnm{Livingstone},~\bfnm{Samuel}\binits{S.}} \AND
  \bauthor{\bsnm{Girolami},~\bfnm{Mark}\binits{M.}}
(\byear{2014}).
\btitle{Information-geometric Markov chain Monte Carlo methods using
  diffusions}.
\bjournal{Entropy}
\bvolume{16}
\bpages{3074--3102}.
\end{barticle}
\endbibitem

\bibitem[\protect\citeauthoryear{Livingstone et~al.}{2018}]{supplement}
\begin{barticle}[author]
\bauthor{\bsnm{Livingstone},~\bfnm{Samuel}\binits{S.}},
  \bauthor{\bsnm{Betancourt},~\bfnm{Michael}\binits{M.}},
  \bauthor{\bsnm{Byrne},~\bfnm{Simon}\binits{S.}} \AND
  \bauthor{\bsnm{Girolami},~\bfnm{Mark}\binits{M.}}
(\byear{2018}).
\btitle{Supplement to ``On the Geometric Ergodicity of Hamiltonian Monte
  Carlo"}.
\end{barticle}
\endbibitem

\bibitem[\protect\citeauthoryear{Mattingly, Stuart and
  Higham}{2002}]{mattingly2002ergodicity}
\begin{barticle}[author]
\bauthor{\bsnm{Mattingly},~\bfnm{Jonathan~C}\binits{J.~C.}},
  \bauthor{\bsnm{Stuart},~\bfnm{Andrew~M}\binits{A.~M.}} \AND
  \bauthor{\bsnm{Higham},~\bfnm{Desmond~J}\binits{D.~J.}}
(\byear{2002}).
\btitle{Ergodicity for SDEs and approximations: locally Lipschitz vector fields
  and degenerate noise}.
\bjournal{Stochastic processes and their applications}
\bvolume{101}
\bpages{185--232}.
\end{barticle}
\endbibitem

\bibitem[\protect\citeauthoryear{Metropolis
  et~al.}{1953}]{metropolis1953equation}
\begin{barticle}[author]
\bauthor{\bsnm{Metropolis},~\bfnm{Nicholas}\binits{N.}},
  \bauthor{\bsnm{Rosenbluth},~\bfnm{Arianna~W}\binits{A.~W.}},
  \bauthor{\bsnm{Rosenbluth},~\bfnm{Marshall~N}\binits{M.~N.}},
  \bauthor{\bsnm{Teller},~\bfnm{Augusta~H}\binits{A.~H.}} \AND
  \bauthor{\bsnm{Teller},~\bfnm{Edward}\binits{E.}}
(\byear{1953}).
\btitle{Equation of state calculations by fast computing machines}.
\bjournal{The journal of chemical physics}
\bvolume{21}
\bpages{1087--1092}.
\end{barticle}
\endbibitem

\bibitem[\protect\citeauthoryear{Meyn and Tweedie}{2012}]{meyn2012markov}
\begin{bbook}[author]
\bauthor{\bsnm{Meyn},~\bfnm{Sean~P}\binits{S.~P.}} \AND
  \bauthor{\bsnm{Tweedie},~\bfnm{Richard~L}\binits{R.~L.}}
(\byear{2012}).
\btitle{Markov chains and stochastic stability}.
\bpublisher{Springer Science \& Business Media}.
\end{bbook}
\endbibitem

\bibitem[\protect\citeauthoryear{Neal}{2011}]{neal2011mcmc}
\begin{barticle}[author]
\bauthor{\bsnm{Neal},~\bfnm{Radford~M}\binits{R.~M.}}
(\byear{2011}).
\btitle{MCMC using Hamiltonian dynamics}.
\bjournal{Handbook of Markov Chain Monte Carlo}
\bvolume{2}.
\end{barticle}
\endbibitem

\bibitem[\protect\citeauthoryear{Ollivier}{2009}]{ollivier2009ricci}
\begin{barticle}[author]
\bauthor{\bsnm{Ollivier},~\bfnm{Yann}\binits{Y.}}
(\byear{2009}).
\btitle{Ricci curvature of Markov chains on metric spaces}.
\bjournal{Journal of Functional Analysis}
\bvolume{256}
\bpages{810--864}.
\end{barticle}
\endbibitem

\bibitem[\protect\citeauthoryear{Ottobre et~al.}{2016}]{ottobre2016function}
\begin{barticle}[author]
\bauthor{\bsnm{Ottobre},~\bfnm{Michela}\binits{M.}},
  \bauthor{\bsnm{Pillai},~\bfnm{Natesh~S}\binits{N.~S.}},
  \bauthor{\bsnm{Pinski},~\bfnm{Frank~J}\binits{F.~J.}},
  \bauthor{\bsnm{Stuart},~\bfnm{Andrew~M}\binits{A.~M.}} \betal{et~al.}
(\byear{2016}).
\btitle{A function space HMC algorithm with second order Langevin diffusion
  limit}.
\bjournal{Bernoulli}
\bvolume{22}
\bpages{60--106}.
\end{barticle}
\endbibitem

\bibitem[\protect\citeauthoryear{Roberts and
  Rosenthal}{1997}]{roberts1997geometric}
\begin{barticle}[author]
\bauthor{\bsnm{Roberts},~\bfnm{Gareth~O}\binits{G.~O.}} \AND
  \bauthor{\bsnm{Rosenthal},~\bfnm{Jeffrey~S}\binits{J.~S.}}
(\byear{1997}).
\btitle{Geometric ergodicity and hybrid Markov chains}.
\bjournal{Electron. Comm. Probab}
\bvolume{2}
\bpages{13--25}.
\end{barticle}
\endbibitem

\bibitem[\protect\citeauthoryear{Roberts and
  Rosenthal}{2004}]{roberts2004general}
\begin{barticle}[author]
\bauthor{\bsnm{Roberts},~\bfnm{Gareth~O}\binits{G.~O.}} \AND
  \bauthor{\bsnm{Rosenthal},~\bfnm{Jeffrey~S}\binits{J.~S.}}
(\byear{2004}).
\btitle{General state space Markov chains and MCMC algorithms}.
\bjournal{Probability Surveys}
\bvolume{1}
\bpages{20--71}.
\end{barticle}
\endbibitem

\bibitem[\protect\citeauthoryear{Roberts and
  Tweedie}{1996a}]{roberts1996geometric}
\begin{barticle}[author]
\bauthor{\bsnm{Roberts},~\bfnm{Gareth~O}\binits{G.~O.}} \AND
  \bauthor{\bsnm{Tweedie},~\bfnm{Richard~L}\binits{R.~L.}}
(\byear{1996}a).
\btitle{Geometric convergence and central limit theorems for multidimensional
  Hastings and Metropolis algorithms}.
\bjournal{Biometrika}
\bvolume{83}
\bpages{95--110}.
\end{barticle}
\endbibitem

\bibitem[\protect\citeauthoryear{Roberts and
  Tweedie}{1996b}]{roberts1996exponential}
\begin{barticle}[author]
\bauthor{\bsnm{Roberts},~\bfnm{Gareth~O}\binits{G.~O.}} \AND
  \bauthor{\bsnm{Tweedie},~\bfnm{Richard~L}\binits{R.~L.}}
(\byear{1996}b).
\btitle{Exponential convergence of Langevin distributions and their discrete
  approximations}.
\bjournal{Bernoulli}
\bpages{341--363}.
\end{barticle}
\endbibitem

\bibitem[\protect\citeauthoryear{Seiler, Rubinstein-Salzedo and
  Holmes}{2014}]{seiler2014positive}
\begin{binproceedings}[author]
\bauthor{\bsnm{Seiler},~\bfnm{Christof}\binits{C.}},
  \bauthor{\bsnm{Rubinstein-Salzedo},~\bfnm{Simon}\binits{S.}} \AND
  \bauthor{\bsnm{Holmes},~\bfnm{Susan}\binits{S.}}
(\byear{2014}).
\btitle{Positive Curvature and Hamiltonian Monte Carlo}.
In \bbooktitle{Advances in Neural Information Processing Systems}
\bpages{586--594}.
\end{binproceedings}
\endbibitem

\bibitem[\protect\citeauthoryear{Stuart}{2010}]{stuart2010inverse}
\begin{barticle}[author]
\bauthor{\bsnm{Stuart},~\bfnm{Andrew~M}\binits{A.~M.}}
(\byear{2010}).
\btitle{Inverse problems: a Bayesian perspective}.
\bjournal{Acta Numerica}
\bvolume{19}
\bpages{451--559}.
\end{barticle}
\endbibitem

\bibitem[\protect\citeauthoryear{Tierney}{1994}]{tierney1994markov}
\begin{barticle}[author]
\bauthor{\bsnm{Tierney},~\bfnm{Luke}\binits{L.}}
(\byear{1994}).
\btitle{Markov chains for exploring posterior distributions}.
\bjournal{the Annals of Statistics}
\bpages{1701--1728}.
\end{barticle}
\endbibitem

\bibitem[\protect\citeauthoryear{Tierney}{1998}]{tierney1998note}
\begin{barticle}[author]
\bauthor{\bsnm{Tierney},~\bfnm{Luke}\binits{L.}}
(\byear{1998}).
\btitle{A note on Metropolis-Hastings kernels for general state spaces}.
\bjournal{Annals of applied probability}
\bpages{1--9}.
\end{barticle}
\endbibitem

\end{thebibliography}


\begin{thebibliography}{}

\bibitem[\protect\citeauthoryear{Livingstone et al.}{2019}]{Livingstone2019}
Livingstone, S., Betancourt, M., Byrne, S. and Girolami, M. (2019).
On the geometric ergodicity of Hamiltonian Monte Carlo.
\textit{Bernoulli} \textbf{25} 3109--3138.
\href{https://doi.org/10.3150/18-BEJ1083}{doi:10.3150/18-BEJ1083}.

\bibitem{gruffaz2026theoretical}
Gruffaz, S., Kim, K., Guehtar, F., Duval-Decaix, H. and Trautmann, P. (2026).
A theoretical comparison of No-U-Turn Sampler variants: Necessary and sufficient
convergence conditions and mixing time analysis under Gaussian targets.
Preprint. Available at arXiv:2603.18640.

\end{thebibliography}


\appendix

\sname{Supplement A}

\section{Examples of $\pi$-irreducibility}

The below example shows how the Hamiltonian Monte Carlo proposal transition can produce a method which is not $\pi$-irreducible, and hence will not be ergodic.

\begin{ex}  \label{ex:irr}
Take $\pi(x) \propto e^{-x^2/2}$, meaning $\nabla U(x) = x$, and set $L = 2$.  Then the HMC proposal becomes
\begin{align*}
x_{2\e} 
&= x_0 - \e^2 x_0 - \e^2 (x_0 - \e^2 x_0 + \e p_0) + 2\e p_0, \\
&= (1 - 2\e^2 + \e^4)x_0 + (2\e - \e^3)p_0.
\end{align*}
Setting $\e = \sqrt{2}$ means $2\e - \e^3 = 0$, so that
\[
x_{2\e} = (1 - 4 + 4)x_0 = x_0.
\]
With this transition the proposal kernel is simply $Q(x,\cdot) = \delta_x(\cdot)$, so the chain is not $\pi$-irreducible unless $\pi(\cdot) = \delta_x(\cdot)$.
\end{ex}

The following diagram give intuition for the $\pi$-irreducibility argument of the idealised Hamiltonian Monte Carlo method.

\vspace{-0.3cm}
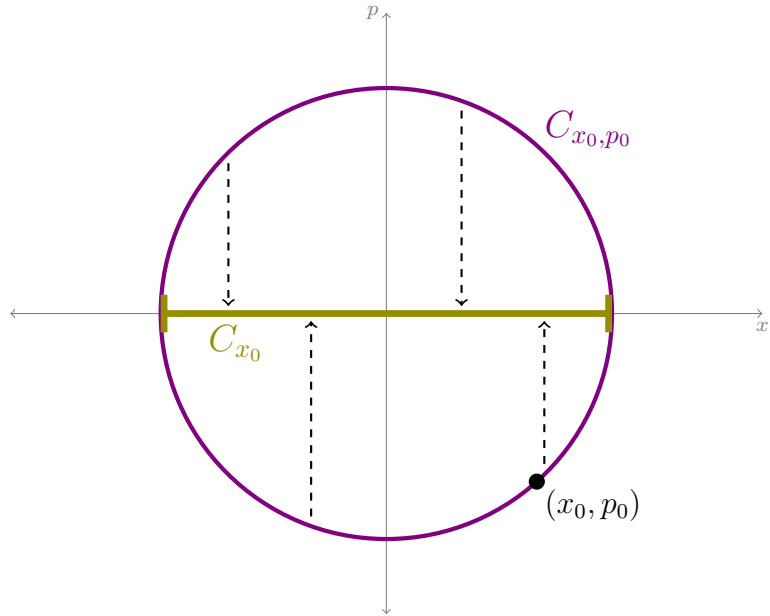
\begin{figure}[H]
\begin{center}
\begin{tikzpicture}
\draw [<->, thin, gray] (0,4) node [left] {$p$} -- (0,0) -- 
(5,0) node [below] {$x$};
\draw [<->, thin, gray] (0,-4) -- (0,0) -- (-5,0);
\draw [violet, ultra thick] (0,0) circle [radius=3];
\draw [|-|, olive, line width=0.09cm] (-3,0) -- (3,0);
\draw [dashed, thick, ->] (1,2.7) -- (1,0.1);
\draw [dashed, thick, ->] (-2.1,2) -- (-2.1,0.1);
\draw [dashed, thick, ->] (-1,-2.7) -- (-1,-0.1);
\draw [dashed, thick, ->] (2.1,-2) -- (2.1,-0.1);
\draw[fill] (2,-2.236) circle [radius=0.1];
\node [below right] at (2,-2.236) {\large $(x_0,p_0)$};
\node [above right, violet] at (2, 2.1) {\Large $C_{x_0,p_0}$};
\node [below, olive] at (-2,-0.05) {\Large $C_{x_0}$};
\end{tikzpicture}
\end{center}
\caption{\label{fig:hmcperiod} The contour $C_{x_0,p_0} = \{ (y,z) \in \R^2 : y^2 + z^2 = 9\}$ for the Hamiltonian flow with Gaussian target $\pi(x) \propto e^{-x^2/2}$, with current point $(x_0,p_0)$ lying on the circle of radius 3, and its projection onto the set $C_{x_0} = [-3,3]$.}
\end{figure}

\section{Connections between HMC and Langevin dynamics}

Hamiltonian Monte Carlo is based on interspersing Hamiltonian dynamics given by the equations
\[
\dot{x} = p, ~~ \dot{p} = -\nabla U(x)
\]
with intermittent re-sampling of $p \sim N(0,I)$ to inject stochasticity into the system. After a small time period the $x$-coordinate will be
\[
x_T = x_0 - \int_0^T\int_0^t \nabla U(x_s)dsdt  + Tp_0
\]
Taking $\sqrt{\delta t} \ll 1$ then for suitably regular $\nabla U(x)$ one can approximate this with the expression
\begin{align*}
x_{\sqrt{\delta t}} &\approx x_0 - \int_0^{\sqrt{\delta t}} t\nabla U(x_0) dt  + \sqrt{\delta t} p_0, \\
&= x_0 - (1/2)\nabla U(x_0)\delta t  + \sqrt{\delta t}p_0.
\end{align*}
Hence, if the dynamics are only performed for a short period before momentum re-sampling, the dynamics will be close to those of an overdamped Langevin diffusion described by the stochastic differential equation
\[
dX_t = -(1/2)\nabla U(X_t)dt + dW_t.
\]
If instead $T$ is typically large in between momentum refreshments, then one can instead consider underdamped Langevin dynamics, described by the system
\begin{align*}
dX_t &= V_t dt, \\
dV_t &= -\nabla U(X_t)dt - \gamma V_t dt + \sqrt{2\gamma}dW_t,
\end{align*}
for some $\gamma>0$. This system more obviously relates to HMC, since it consists of a Hamiltonian part combined with some stochasticity, which is also injected into the momentum variable $V_t$. The stochasticity is however in this case continuously injected in the form of an Ornstein--Uhlenbeck (OU) process. As such, as the dynamics evolve the conservative Hamiltonian flow is constantly perturbed by small random adjustments to the momentum. This is closely connected to the behaviour of generalized HMC, in which integration times are typically smaller and the momentum is only partially refreshed using $p_{new} \sim N(\xi p_{old}, (1-\xi^2)I)$. Indeed, if $\xi := e^{-\gamma t}$ then this step is an exact solution to the OU part of the system.

\section{Proof of inwards convergence for the Exponential family model class}

We provide a verbose proof in the case $\beta\in(1,2)$, to elaborate on the short version provided in the main text.  In what follows $U(x):=\alpha|x|^{\beta}$ for some $\alpha>0$,
and $U^{(k)}(x):=d^{k}U(x)/dx^{k}.$ We precede the main result with
a technical lemma.
\begin{lem}
\label{lem:distance}For every $i\in\{1,...,L\}$, if $p_{0}=o(x_{0}^{\beta-1})$
then

\begin{align*}
(x_{i\varepsilon}-x_{0}) & =-\frac{i^{2}\varepsilon^{2}}{2}U'(x_{0})+\frac{\varepsilon^{4}}{2}\sum_{j=1}^{i-1}(i-j)j^{2}U^{(2)}(x_{0})U'(x_{0})-i\varepsilon p_{0}+o(x_{0}^{2\beta-3}).
\end{align*}
\end{lem}
\begin{proof}
Direct calculation gives
\begin{align*}
(x_{i\varepsilon}-x_{0}) & =-\frac{i\varepsilon^{2}}{2}U'(x_{0})-\varepsilon^{2}\sum_{j=1}^{i-1}(i-j)U'(x_{j\varepsilon})-i\varepsilon p_{0}\\
 & =-\frac{i\varepsilon^{2}}{2}U'(x_{0})-\varepsilon^{2}\sum_{j=1}^{i-1}(i-j)\left[U'(x_{0})+\sum_{k=1}^{\infty}U^{(k+1)}(x_{0})(x_{j\varepsilon}-x_{0})^{k}\frac{1}{k!}\right]-i\varepsilon p_{0}\\
 & =-\frac{i^{2}\varepsilon^{2}}{2}U'(x_{0})+\frac{\varepsilon^{4}}{2}\sum_{j=1}^{i-1}(i-j)j^{2}U^{(2)}(x_{0})U'(x_{0})-i\varepsilon p_{0}+o(x_{0}^{2\beta-3}).
\end{align*}
\end{proof}
\begin{prop}
If $\beta\in(1,2)$ then HMC converges inwards.
\end{prop}
\begin{proof}
We have
\begin{align*}
K(p_{0})-K(p_{L\varepsilon}) & =\frac{1}{2}p_{0}^{2}-\frac{1}{2}\left(p_{0}-\frac{\varepsilon}{2}(U'(x_{0})+U'(x_{L\varepsilon}))-\varepsilon\sum_{i=1}^{L-1}U'(x_{i\varepsilon})\right)^{2}\\
 & =\frac{\varepsilon}{2}p_{0}\left[U'(x_{0})+U'(x_{L\varepsilon})+2\sum_{i=1}^{L-1}U'(x_{i\varepsilon})\right]\\
 & \qquad -\frac{\varepsilon^{2}}{8}\left[U'(x_{0})+U'(x_{L\varepsilon})+2\sum_{i=1}^{L-1}U'(x_{i\varepsilon})\right]^{2}.
\end{align*}

Take $p_{0}=o(x_{0}^{\beta-1})$, and note that this occurs with probability
one as $x_{0}\to\infty$. We write $K(p_{0})-K(p_{L\varepsilon})=\kappa_{1}+\kappa_{2}$
where $\kappa_{1}:=\varepsilon p_{0}\left[U'(x_{0})+U'(x_{L\varepsilon})+2\sum_{i=1}^{L-1}U'(x_{i\varepsilon})\right]/2$
and $\kappa_{2}:=-\varepsilon^{2}\left[U'(x_{0})+U'(x_{L\varepsilon})+2\sum_{i=1}^{L-1}U'(x_{i\varepsilon})\right]^{2}/8$.
Then
\begin{align*}
\kappa_{1} & =\frac{\varepsilon}{2}p_{0}\left[U'(x_{0})+U'(x_{L\varepsilon})+2\sum_{i=1}^{L-1}U'(x_{i\varepsilon})\right]\\
 & =\frac{\varepsilon}{2}p_{0}\left[2LU'(x_{0})+\sum_{k=1}^{\infty}U^{(k+1)}(x_{0})(x_{L\varepsilon}-x_{0})^{k}\frac{1}{k!}+2\sum_{k=1}^{\infty}\frac{1}{k!}U^{(k+1)}(x_{0})\sum_{i=1}^{L-1}(x_{i\varepsilon}-x_{0})^{k}\right]\\
 & =L\varepsilon p_{0}U'(x_{0})+o(x^{3\beta-4}).
\end{align*}
And similarly up to $o(x_{0}^{3\beta-4})$ terms
\begin{align*}
\kappa_{2} & =-\frac{\varepsilon^{2}}{8}\left[U'(x_{0})+U'(x_{L\varepsilon})+2\sum_{i=1}^{L-1}U'(x_{i\varepsilon})\right]^{2}\\
 & =-\frac{\varepsilon^{2}}{8}\left[2LU'(x_{0})+\sum_{k=1}^{\infty}U^{(k+1)}(x_{0})(x_{L\varepsilon}-x_{0})^{k}\frac{1}{k!}+2\sum_{k=1}^{\infty}\frac{1}{k!}U^{(k+1)}(x_{0})\sum_{i=1}^{L-1}(x_{i\varepsilon}-x_{0})^{k}\right]^{2}\\
 & =-\frac{\varepsilon^{2}}{8}\left[4L^{2}U'(x_{0})^{2}+4LU'(x_{0})U^{(2)}(x)(x_{L\varepsilon}-x_{0})+8LU'(x_{0})U^{(2)}(x_{0})\sum_{i=1}^{L-1}(x_{i\varepsilon}-x_{0})\right].
\end{align*}
Since $(x_{L\varepsilon}-x_{0})=-L^{2}\varepsilon^{2}U'(x_{0})/2+o(x^{\beta-1})$
then
\begin{align*}
\kappa_{2} & =-\frac{\varepsilon^{2}}{8}\left(4L^{2}U'(x_{0})^{2}-4L\frac{L^{2}\varepsilon^{2}}{2}U'(x_{0})^{2}U^{(2)}(x_{0})-8L\frac{\varepsilon^{2}}{2}U'(x_{0})^{2}U^{(2)}(x_{0})\sum_{i=1}^{L-1}i^{2}\right)\\
 & =-\frac{\varepsilon^{2}}{8}\left(4L^{2}U'(x_{0})^{2}-\varepsilon^{2}\left[2L^{3}+4L\sum_{i=1}^{L-1}i^{2}\right]U'(x_{0})^{2}U^{(2)}(x_{0})+...\right)+o(x_{0}^{3\beta-4}).
\end{align*}
Now turning to $U(x_{0})-U(x_{L\varepsilon})$ and using Lemma \ref{lem:distance}
we have
\begin{align*}
U(x_{0})-U(x_{L\varepsilon}) & =-U'(x_{0})(x_{L\varepsilon}-x_{0})-\frac{1}{2}U^{(2)}(x_{0})(x_{L\varepsilon}-x_{0})^{2}+o(x^{3\beta-4})\\
 & =\frac{L^{2}\varepsilon^{2}}{2}U'(x_{0})^{2}-L\varepsilon p_{0}U'(x_{0})\\
 & \qquad -\varepsilon^{4}\left[\frac{L^{4}}{8}+\frac{1}{2}\sum_{i=1}^{L-1}(L-i)i^{2}\right]U^{(2)}(x_{0})U'(x_{0})^{2}+o(x^{3\beta-4}).
\end{align*}
So combining gives up to $o\left(x_{0}^{3\beta-4}\right)$ terms then $H(x_{0},p_{0})-H(x_{L\varepsilon},p_{L\varepsilon})$ is
\[
\varepsilon^{4}\left[\frac{1}{4}L^{3}-\frac{1}{8}L^{4}+\frac{1}{2}\sum_{i=1}^{L-1}Li^{2}-\frac{1}{2}\left(\sum_{i=1}^{L-1}(L-i)i^{2}\right)\right]U'(x_{0})^{2}U^{(2)}(x_{0}).
\]
The coefficient of the leading order term divided by $\varepsilon^{4}$
is $L^{3}/4-L^{4}/8+\sum_{i=1}^{L-1}(L-(L-i))i^{2}/2=L^{2}/8$. Since
this is $>0$ then as $x_{0}\to\infty$ the result is proved. An analogous
argument holds as $x_{0}\to-\infty$.
\end{proof}


\section{Corrigendum related to Theorem \ref{thm:light} and Corollary \ref{crly:expfam}}

\makeatletter\let\section\specialsection\setcounter{section}{0}\gdef\thesection{\arabic{section}}\gdef\section@prefix{}\gdef\section@numbersep{.}\makeatother




\startlocaldefs
\theoremstyle{plain}
\newtheorem*{lemma515}{Lemma 5.15}
\newtheorem*{lemma516}{Lemma 5.16}
\newtheorem*{lemma517}{Lemma 5.17}
\newtheorem*{theorem514}{Theorem 5.14}
\newtheorem*{corollary23}{Corollary 2.3(ii), $\beta>2$ case}
\endlocaldefs


\begin{frontmatter}

\title{Corrigendum to ``On the geometric ergodicity of Hamiltonian Monte Carlo''}
\runtitle{Corrigendum: geometric ergodicity of HMC}

\begin{aug}
\author[A]{\inits{S.}\fnms{Samuel}~\snm{Livingstone}\ead[label=e1]{samuel.livingstone@ucl.ac.uk}}
\author[B]{\inits{M.}\fnms{Michael}~\snm{Betancourt}\ead[label=e2]{betan@symplectomorphic.com}}
\author[C]{\inits{S.}\fnms{Simon}~\snm{Byrne}\ead[label=e3]{simonbyrne@gmail.com}}
\author[D]{\inits{M.}\fnms{Mark}~\snm{Girolami}\ead[label=e4]{mag92@cam.ac.uk}}

\address[A]{Department of Statistical Science, University College London,
Gower Street, London WC1E 6BT, United Kingdom\printead[presep={,\ }]{e1}}
\address[B]{Symplectomorphic LLC, New York, NY, USA\printead[presep={,\ }]{e2}}
\address[C]{NVIDIA Corporation, 2788 San Tomas Expressway,
Santa Clara, CA 95051, USA\printead[presep={,\ }]{e3}}
\address[D]{Department of Engineering, University of Cambridge, Trumpington Street, Cambridge CB2 1PZ, United Kingdom\printead[presep={,\ }]{e4}}

\runauthor{S. Livingstone et al.}
\end{aug}

\begin{abstract}
We correct a consequential typographical error in Theorem 5.14 of \cite{Livingstone2019}. Condition (18) should require
$\|y\|\geq 3\|x\|$, rather than $\|y\|\geq 2\|x\|$, while retaining the
requirement $\|\nabla U(y)\|\geq 3\|\nabla U(x)\|$. The original condition is
not satisfied by the exponential family $E(\beta,\alpha)$ for every
$\beta>2$, meaning Corollary 2.3 was incomplete as stated. We give replacement proofs of Lemmas 5.15--5.17 and
Theorem 5.14 for completeness.  For
$E(\beta,\alpha)$ the corrected condition holds for every $\beta>2$, meaning Corollary 2.3 remains valid as stated.
\end{abstract}

\begin{keyword}
\kwd{geometric ergodicity}
\kwd{Hamiltonian Monte Carlo}
\kwd{leapfrog integrator}
\end{keyword}

\end{frontmatter}

\section{Details of the typographical error and correction}

In \cite{Livingstone2019}, p.~3128, condition (18), used in Theorem 5.14 and
then in part (i) of Theorem 2.2, states that there is a fixed $C<\infty$ such
that whenever $\|y\|\geq2\|x\|\geq C$ then
$\|\nabla U(y)\|\geq3\|\nabla U(x)\|$. For the one-dimensional exponential
family $E(\beta,\alpha)$, for which $U(x)=\alpha|x|^\beta+\mathrm{const.}$ in
the tails,
\[
 \frac{\|\nabla U(y)\|}{\|\nabla U(x)\|}
 \geq 2^{\beta-1}
\]
whenever $|y|\geq2|x|$ and both points are in the tails.  The
printed condition is therefore guaranteed only when
$\beta\geq1+\log(3)/\log(2)$, rather than for every $\beta>2$. The proof
of Corollary 2.3(ii) therefore has a gap for
\[
 2<\beta<1+\frac{\log 3}{\log 2}.
\]
Condition (18) should instead read: there is a fixed $C<\infty$ such that
whenever $\|y\|\geq3\|x\|\geq C$ then
\begin{equation}
 \|\nabla U(y)\|\geq3\|\nabla U(x)\|.
 \tag{18}
 \label{eq:18}
\end{equation}
Condition (19) can be left unchanged. With these updates Corollary 2.3(ii) is true as stated, lack of geometric ergodicity is proven for any $\beta > 2$.  The conclusion of Theorem 5.14, and therefore
the conclusion used in Theorem 2.2, is unchanged. We thank the authors of \citep{gruffaz2026theoretical} for bringing the error to our attention.

There is another (albeit inconsequential) typographical error in the printed proof.  Lemma 5.16 is stated under
condition (17) alone, but its proof invokes condition (18). The following
minimally-changed replacement proofs show directly that condition (17) alone is sufficient for the lemma.

\section{Replacement proofs}

Fix a leapfrog step size $\varepsilon>0$. The first leapfrog position update
is
\begin{equation}
 x_{\varepsilon}
 =x_0-\frac{\varepsilon^2}{2}\nabla U(x_0)+\varepsilon p_0.
 \label{eq:firststep}
\end{equation}
For each $i\geq1$, eliminating the intermediate half-step momentum from two
successive leapfrog position updates gives the exact recurrence
\begin{equation}
 x_{(i+1)\varepsilon}
 =2x_{i\varepsilon}-x_{(i-1)\varepsilon}
 -\varepsilon^2\nabla U(x_{i\varepsilon}).
 \label{eq:recurrence}
\end{equation}
We retain condition (17) of \cite{Livingstone2019}, namely
\begin{equation}
 \lim_{\|x\|\to\infty}
 \frac{\|\nabla U(x)\|}{\|x\|}=\infty.
 \tag{17}
 \label{eq:17}
\end{equation}

\begin{lemma515}
If \eqref{eq:17} holds, then there exists an $\eta<\infty$ such that for all
$\|x_0\|>\eta$ and any $\|p_0\|\leq\|x_0\|^\delta$ for any fixed $\delta<1$, it
holds that
\[
 \|x_{\varepsilon}\|>3\|x_0\|.
\]
\end{lemma515}

\begin{proof}
Taking norms in \eqref{eq:firststep} gives
\[
 \|x_{\varepsilon}\|
 \geq
 \frac{\varepsilon^2}{2}\|\nabla U(x_0)\|
 -\|x_0\|-\varepsilon\|p_0\|.
\]
Dividing by $\|x_0\|$ gives
\begin{equation}
 \frac{\|x_{\varepsilon}\|}{\|x_0\|}
 \geq
 \frac{\varepsilon^2}{2}
 \frac{\|\nabla U(x_0)\|}{\|x_0\|}
 -1-\varepsilon\frac{\|p_0\|}{\|x_0\|}.
 \label{eq:firststepbound}
\end{equation}
The first term on the right-hand side tends to infinity by \eqref{eq:17},
while the last term is bounded above by
$\varepsilon\|x_0\|^{\delta-1}$ and tends to zero. Hence the right-hand side
of \eqref{eq:firststepbound} is larger than $3$ for all sufficiently large
$\|x_0\|$.
\end{proof}

\begin{lemma516}
If \eqref{eq:17} holds, then there exists an $\eta<\infty$ such that for all
$\|x_0\|>\eta$ and any $\|p_0\|\leq\|x_0\|^\delta$ for any fixed $\delta<1$, it
holds that
\begin{equation}
 \|x_{i\varepsilon}\|\geq3\|x_{(i-1)\varepsilon}\|,
 \qquad i=1,\ldots,L,
 \label{eq:stepwiseexpansion}
\end{equation}
meaning
\begin{equation}
 \|x_{L\varepsilon}\|\geq3^L\|x_0\|.
 \label{eq:3Lexpansion}
\end{equation}
\end{lemma516}

\begin{proof}
By \eqref{eq:17}, there is an $R<\infty$ such that
\begin{equation}
 \varepsilon^2\frac{\|\nabla U(x)\|}{\|x\|}
 \geq\frac{16}{3}
 \qquad\text{whenever }\|x\|\geq R.
 \label{eq:gradientthreshold}
\end{equation}
The claim holds once $\|x_0\|$ is large enough. Choose $x_0$ such that $\|x_0\| \geq \max(R,\eta)$, meaning Lemma 5.15 applies and $\|x_0\|\geq R$.
The case $i=1$ then follows from Lemma 5.15.  Now suppose, for some $i\in\{1,\ldots,L-1\}$, that
$\|x_{i\varepsilon}\|\geq3\|x_{(i-1)\varepsilon}\|$. Taking norms in
\eqref{eq:recurrence} gives
\begin{align*}
 \frac{\|x_{(i+1)\varepsilon}\|}{\|x_{i\varepsilon}\|}
 &\geq
 \varepsilon^2
 \frac{\|\nabla U(x_{i\varepsilon})\|}{\|x_{i\varepsilon}\|}
 -2-
 \frac{\|x_{(i-1)\varepsilon}\|}{\|x_{i\varepsilon}\|}\\
 &\geq
 \varepsilon^2
 \frac{\|\nabla U(x_{i\varepsilon})\|}{\|x_{i\varepsilon}\|}
 -2-\frac13.
\end{align*}
The preceding steps of the induction imply
$\|x_{i\varepsilon}\|\geq\|x_0\|\geq R$, and therefore
\eqref{eq:gradientthreshold} gives
\[
 \frac{\|x_{(i+1)\varepsilon}\|}{\|x_{i\varepsilon}\|}
 \geq\frac{16}{3}-\frac{7}{3}=3.
\]
This proves \eqref{eq:stepwiseexpansion} by induction, and
\eqref{eq:3Lexpansion} follows immediately.
\end{proof}

We also retain condition (19): for the fixed number $L$ of leapfrog steps,
\begin{equation}
 \lim_{\substack{\|x\|\to\infty\\
                  \|y\|\geq2^L\|x\|}}
 \left(U(x)-U(y)-\frac12\|x\|^2\right)=-\infty.
 \tag{19}
 \label{eq:19}
\end{equation}

\begin{lemma517}
If \eqref{eq:17}, \eqref{eq:18} and \eqref{eq:19} hold, then for any
$\delta<1$ it holds that
\[
 \lim_{\substack{\|x_0\|\to\infty\\
                  \|p_0\|\leq\|x_0\|^\delta}}
 \alpha(x_0,x_{L\varepsilon})=0.
\]
\end{lemma517}

\begin{proof}
Take $\|x_0\|$ sufficiently large that Lemma 5.16 applies and that
$3\|x_{i\varepsilon}\|\geq C$ for $i=0,\ldots,L-1$. Using
\eqref{eq:stepwiseexpansion} and the corrected condition (18) then gives for $i \in \{0,...,L-1\}$
\begin{equation}
 \|\nabla U(x_{i\varepsilon})\|
 \leq3^{-(L-i)}\|\nabla U(x_{L\varepsilon})\|.
 \label{eq:gradientgeometric}
\end{equation}
Now, recall from equation (11) of \cite{Livingstone2019} that
\[
 p_{L\varepsilon}
 =p_0-\frac{\varepsilon}{2}\nabla U(x_0)
 -\varepsilon\sum_{i=1}^{L-1}\nabla U(x_{i\varepsilon})
 -\frac{\varepsilon}{2}\nabla U(x_{L\varepsilon}).
\]
Using \eqref{eq:gradientgeometric} and the reverse triangle inequality,
\begin{align*}
 \|p_{L\varepsilon}\|
 &\geq
 \frac{\varepsilon}{2}\|\nabla U(x_{L\varepsilon})\|
 -\varepsilon\sum_{i=1}^{L-1}\|\nabla U(x_{i\varepsilon})\|
 -\frac{\varepsilon}{2}\|\nabla U(x_0)\|-\|p_0\|\\
 &\geq
 \frac{\varepsilon}{2}\|\nabla U(x_{L\varepsilon})\|
 -\varepsilon\sum_{i=1}^{L}\left(\frac13\right)^i
 \|\nabla U(x_{L\varepsilon})\|-\|p_0\|\\
 &=
 \frac{\varepsilon}{2}
 \left\{1-2\sum_{i=1}^{L}\left(\frac13\right)^i\right\}
 \|\nabla U(x_{L\varepsilon})\|-\|p_0\|\\
 &=
 \frac{\varepsilon}{2\cdot3^L}
 \|\nabla U(x_{L\varepsilon})\|-\|p_0\|.
\end{align*}
Set $\gamma_L=3^{-L}$. Since \eqref{eq:17} holds and
$\|x_{L\varepsilon}\|\geq3^L\|x_0\|$, for all sufficiently large
$\|x_0\|$ the right-hand side above is positive. Squaring then gives
\begin{align*}
 \|p_{L\varepsilon}\|^2
 &\geq
 \left(
 \frac{\varepsilon\gamma_L}{2}\|\nabla U(x_{L\varepsilon})\|
 -\|p_0\|
 \right)^2\\
 &=
 \frac{\varepsilon^2\gamma_L^2}{4}
 \|\nabla U(x_{L\varepsilon})\|^2
 +\|p_0\|^2
 -\varepsilon\gamma_L
 \|\nabla U(x_{L\varepsilon})\|\|p_0\|.
\end{align*}
It follows that
\begin{equation}
 \|p_0\|^2-\|p_{L\varepsilon}\|^2
 \leq
 \varepsilon\gamma_L\|\nabla U(x_{L\varepsilon})\|
 \left(
 \|p_0\|-\frac{\varepsilon\gamma_L}{4}
 \|\nabla U(x_{L\varepsilon})\|
 \right).
 \label{eq:kineticbound}
\end{equation}
Again using \eqref{eq:17} and \eqref{eq:3Lexpansion}, uniformly over
$\|p_0\|\leq\|x_0\|^\delta$ we may take $\|x_0\|$ sufficiently large that
\[
 \frac{\varepsilon\gamma_L}{4}
 \|\nabla U(x_{L\varepsilon})\|\geq2\|x_0\|,
 \qquad
 \|p_0\|\leq\|x_0\|.
\]
The right-hand side of \eqref{eq:kineticbound} is then at most
$-\|x_0\|^2$, which implies
\[
 U(x_0)-U(x_{L\varepsilon})
 +\frac12\|p_0\|^2-\frac12\|p_{L\varepsilon}\|^2
\leq
 U(x_0)-U(x_{L\varepsilon})-\frac12\|x_0\|^2.
\]
By \eqref{eq:3Lexpansion},
$\|x_{L\varepsilon}\|\geq3^L\|x_0\|\geq2^L\|x_0\|$, so the right-hand
side tends to $-\infty$ by \eqref{eq:19}. Recalling equation (10) of
\cite{Livingstone2019}, this proves
$\alpha(x_0,x_{L\varepsilon})\to0$ uniformly over
$\|p_0\|\leq\|x_0\|^\delta$.
\end{proof}

\begin{theorem514}
If \eqref{eq:17}, \eqref{eq:18} and \eqref{eq:19} hold, then the Hamiltonian
Monte Carlo method with fixed integration time $T=L\varepsilon$ does not
produce a geometrically ergodic Markov chain for any choice $T>0$.
\end{theorem514}

\begin{proof}
Fix any $\delta<1$. From Lemma 5.17,
\[
 \sup_{\|p_0\|\leq\|x_0\|^\delta}
 \alpha(x_0,x_{L\varepsilon})\longrightarrow0
 \qquad\text{as }\|x_0\|\to\infty.
\]
Moreover,
\[
 \mu_G\!\left(\left\{p_0:\|p_0\|>\|x_0\|^\delta\right\}\right)
 \longrightarrow0.
\]
Therefore, writing $r(x_0)$ for the rejection probability as in Definition
3.2 of \cite{Livingstone2019},
\begin{align*}
 1-r(x_0)
 &=\int\alpha(x_0,x_{L\varepsilon})\,\mu_G(dp_0)\\
 &\leq
 \sup_{\|p_0\|\leq\|x_0\|^\delta}
 \alpha(x_0,x_{L\varepsilon})
 +\mu_G\!\left(\left\{p_0:\|p_0\|>\|x_0\|^\delta\right\}\right)
 \longrightarrow0.
\end{align*}
This implies that the rejection probability $r(x_0)\to1$ as $\|x_0\|\to\infty$, and Proposition 3.5 of
\cite{Livingstone2019} then implies that the chain is not geometrically ergodic.
\end{proof}

When the number of leapfrog steps is drawn from a position-independent
distribution $\mathcal{L}(\cdot)$ as in Assumption A1 of \cite{Livingstone2019}, the
same conclusion follows provided the above argument is applied conditional on
each fixed $L$. For each fixed $L$ the conditional acceptance probability
tends to zero, while it is bounded by one. Since the law of $L$ does not
depend on $x_0$, dominated convergence with respect to $\mathcal{L}(\cdot)$ shows that
the unconditional acceptance probability also tends to zero.

\section{Correction to Corollary 2.3}

On p.~3127, the statement that Corollary 2.3(ii) is a direct consequence of
Theorem 2.2 requires the following verification in the case $\beta>2$.

\begin{corollary23}
For the exponential family $E(\beta,\alpha)$ of \cite{Livingstone2019}, the
Hamiltonian Monte Carlo method does not produce a geometrically ergodic
Markov chain for every $\beta>2$ under Assumption A1.
\end{corollary23}

\begin{proof}
For some $M<\infty$, when $|x|> M$ we have by assumption that
$\pi(x)\propto\exp(-\alpha|x|^\beta)$,
so
\[
 \nabla U(x)=\alpha\beta\operatorname{sgn}(x)|x|^{\beta-1}.
\]
When $\beta>2$,
\[
 \frac{\|\nabla U(x)\|}{\|x\|}
 =\alpha\beta|x|^{\beta-2}\longrightarrow\infty,
\]
which is condition \eqref{eq:17}.

Choose $C>3M$. If $|y|\geq3|x|\geq C$, then both $x$ and $y$ are in the tail
region and
\[
 \frac{\|\nabla U(y)\|}{\|\nabla U(x)\|}
 =\left(\frac{|y|}{|x|}\right)^{\beta-1}
 \geq3^{\beta-1}>3.
\]
The corrected condition (18) therefore holds for every $\beta>2$.

Finally, for each fixed $L\geq1$, if $|y|\geq2^L|x|$ and $|x|$ is
sufficiently large, then
\begin{align*}
 U(x)-U(y)-\frac12|x|^2
 &\leq
 -\alpha\bigl(2^{L\beta}-1\bigr)|x|^\beta
 -\frac12|x|^2
 \longrightarrow-\infty,
\end{align*}
meaning condition \eqref{eq:19} holds for every fixed $L$. The corrected
Theorem 5.14 applies conditional on each fixed $L$, and dominated convergence
with respect to the position-independent distribution $\mathcal{L}(\cdot)$ gives the
corresponding result under Assumption A1. The Hamiltonian Monte Carlo
method therefore does not produce a geometrically ergodic Markov chain for every
$\beta>2$.
\end{proof}

The $\beta<1$ part of Corollary 2.3(ii), and all parts of Corollary 2.3(i),
are unaffected by this correction. Therefore Corollary 2.3 remains valid as
stated. No other changes to the results of \cite{Livingstone2019} are
required.



\end{document}